\documentclass[11pt]{article}
\usepackage{paralist}
\usepackage{color}

\usepackage{amsfonts}
\usepackage{amssymb}

\usepackage{amsmath}
\usepackage{constants}
\usepackage{amsthm}
\usepackage{tikz}
\usetikzlibrary{shapes,calc,positioning,arrows,hobby}
\makeatletter
\newtheorem*{rep@theorem}{\rep@title}
\newcommand{\newreptheorem}[2]{%
\newenvironment{rep#1}[1]{%
 \def\rep@title{#2 \ref{##1}}%
 \begin{rep@theorem}}%
 {\end{rep@theorem}}}
\makeatother

\allowdisplaybreaks

\usepackage{ifpdf}
\newcommand{\mydriver}{hypertex}
\ifpdf
 \renewcommand{\mydriver}{pdftex}
\fi
\usepackage[breaklinks,\mydriver]{hyperref}
\usepackage{subcaption}

\usepackage[margin=1in]{geometry}

\usepackage{tikz}
\usetikzlibrary{shapes,calc,positioning,arrows,}
\newcommand{\ie}{\text{i.\,e.}\ }
\newcommand{\eg}{\emph{e.\,g.}\ }

\newcommand{\ar}{\operatorname{ar}}
\newcommand{\Npos}{\mathbb N_{> 0}}

\newcommand{\rot}{\operatorname{ROT}}
\newcommand{\dist}{\operatorname{dist}}
\newcommand{\indexSetRotation}{[D]^2}
\newcommand{\indexSetH}{([D]^2)^2}

\renewcommand{\phi}{\varphi}

\newcommand*\circled[1]{\tikz[baseline=(char.base)]{
		\node[shape=circle,draw,inner sep=1pt] (char) {#1};}}
\DeclareMathOperator{\zigzag}{\circled{{\rm z}}}

\newcommand{\other}[1]{\tilde{#1}}
\newcommand{\struc}[1]{#1}
\newcommand{\univ}[1]{U(#1)}
\newcommand{\classStruc}[1]{\mathcal{#1}}
\newcommand{\rel}[2]{#1(#2)}
\newcommand{\gaifman}[1]{G(#1)}
\newcommand{\graphProp}{\mathcal{P}_{\operatorname{graph}}}

\newcommand{\vet}[1]{\ensuremath{\overline{#1}}}

\theoremstyle{plain}
\newtheorem{theorem}{Theorem}[section]

\newtheorem{lemma}[theorem]{Lemma}

\newtheorem{claim}[theorem]{Claim}

\newtheorem{observation}[theorem]{Observation}

\newreptheorem{theorem}{Theorem}
\newreptheorem{lemma}{Lemma}
\newtheorem{definition}[theorem]{Definition}

\newtheorem{question}[theorem]{Open Question}
\theoremstyle{definition}

\newtheorem{example}[theorem]{Example}

\newcommand{\junk}[1]{{}}

\title{GSF-locality is not sufficient for proximity-oblivious testing}

\author{Isolde Adler
	\thanks{School of Computing, University of Leeds, UK. Email: \url{I.M.Adler@leeds.ac.uk}. }
	\and
	Noleen K\"ohler
	\thanks{School of Computing, University of Leeds, UK. Email: 	\url{scnk@leeds.ac.uk}.}
	\and
	Pan Peng
	\thanks{Department of Computer Science, University of Sheffield, UK. Email: \url{p.peng@sheffield.ac.uk}.}
}

\date{}

\begin{document}
	
	
\begin{titlepage}
		
\maketitle
		
\thispagestyle{empty}
\begin{abstract}
	In Property Testing, \emph{proximity-oblivious testers (POTs)} form a class of particularly simple testing algorithms, where a basic test is performed a number of times that may depend on the proximity parameter, but the basic test itself is independent of the proximity parameter.

	In their seminal work, Goldreich and Ron [STOC 2009; SICOMP 2011] show that
	the graph properties that allow constant-query proximity-oblivious testing in the bounded-degree model are precisely the properties that can be expressed as a \emph{generalised subgraph freeness (GSF)} property that satisfies the \emph{non-propagation} condition.
	It is left open whether the non-propagation condition is necessary. Indeed, calling
	properties expressible as a generalised subgraph freeness property \emph{GSF-local properties}, they ask
	whether all GSF-local properties are non-propagating.
	We give a negative answer by exhibiting a property of graphs that is GSF-local 
	and propagating. Hence in particular, our property does not
	admit a POT, despite being GSF-local.
	We prove our result by exploiting a recent work of the authors which constructed a first-order (FO) property that is not testable [SODA 2021], and a new connection between FO properties and GSF-local properties via neighbourhood profiles.   

\end{abstract}
		
\end{titlepage}
	
\section{Introduction}

Graph property testing is a framework for studying sampling-based graph algorithms. Given a graph property $\mathcal{P}$, the goal is to design a (randomised) algorithm, called \emph{tester}, that distinguishes between graphs that satisfy $\mathcal{P}$ from those that are `far' from satisfying $\mathcal{P}$, where the notion `being far' depends on the underlying query access model and  is always parametrised by a \emph{proximity parameter $\varepsilon>0$}. The query model also specifies the class of graphs and the types of queries allowed by the algorithm. The two most  well known models for graph property testing are the \emph{dense graph model} and the \emph{bounded-degree graph model} (see \cite{goldreich2017introduction}). Towards an understanding of which graph properties are testable with a constant number of queries in each model, much progress has been made since the framework of property testing was introduced \cite{rubinfeld1996robust,goldreich1998property}. To illustrate, a full characterization of the properties that are testable with a constant number of queries in the dense graph model has been obtained by Alon, Fischer, Newman, and Shapira \cite{alon2009combinatorial}. 

Typical property testers make decisions regarding the global property of the graph from the local views. In the extreme case, a tester could make local views independent of the distance to a predetermined set of graphs. Motivated by this, Goldreich and Ron \cite{goldreich2011proximity} initiated the study of (one-sided error) \emph{proximity-oblivious testers (POTs)} for graphs, where a tester simply repeats a basic test for a number of times that depends on the proximity parameter, and the basic tester is oblivious of the proximity parameter. They gave characterizations of graph properties that can be tested with constant query complexity by a POT in both dense graph model and the bounded-degree model. In each model, it is known that the class of properties that have constant-query POTs is a strict subset of the class of properties that are testable (by standard testers). 

In this paper, we focus on the bounded-degree graph model \cite{GoldreichRon2002}. In this model, the algorithm is given query access to an input graph with maximum degree bounded by $d$, where $d$ is some constant. For any specified query $v$ and an index $i\leq d$, the algorithm can obtain the $i$-th neighbor of $v$ if it exists, and a special symbol $\bot$ otherwise. Given a {proximity parameter} $\varepsilon>0$, an $n$-vertex graph with maximum degree at most $d$ is said to be $\varepsilon$-far from a property $\mathcal{P}$ if one needs to add and/or delete more than $\varepsilon dn$ edges to make it satisfy $\mathcal{P}$. A property is said to be \emph{testable} if there exists a tester that makes only a \emph{constant} number of queries to the input graph $G$, and distinguishes if $G$ satisfies the property $\mathcal{P}$ or is $\varepsilon$-far from satisfying $\mathcal{P}$, with success probability at least $\frac23$. Here the constant is a number that might depend on $\varepsilon$ and $d$, but is independent of the size of the input graph. It has been known that many properties are testable, such as subgraph-freeness, $k$-edge connectivity, cycle-freeness, being Eulerian, degree-regularity~\cite{GoldreichRon2002}, minor-freeness~\cite{benjamini2010every,hassidim2009local,kumar2019random}, hyperfinite properties \cite{NewmanSohler2013}, $k$-vertex connectivity~\cite{yoshida2012property,forster2019computing}, and subdivision-freeness~\cite{kawarabayashi2013testing}.

Turning to POTs, informally, 
a (one-sided error) POT for a property $\mathcal{P}$ is a tester that always accepts a graph $G$ if it satisfies $\mathcal{P}$, and rejects $G$ with probability that is a monotonically increasing function of the distance of $G$ from the property $\mathcal{P}$. We say $\mathcal{P}$ is \emph{proximity-oblivious testable} if such a tester exists for $\mathcal{P}$ with constant query complexity. To characterise the class of proximity-oblivious testable properties in the bounded-degree model, Goldreich and Ron \cite{goldreich2011proximity} introduced a notion of generalized subgraph freeness (GSF), that extends the notions of induced subgraph freeness and (non-induced) subgraph freeness. A graph property is called a \emph{GSF-local} property if it is expressible as a GSF property. It has been shown in \cite{goldreich2011proximity} that a graph property is constant-query proximity-oblivious testable if and only if it is a GSF-local property that satisfies a so-called \emph{non-propagation} condition. 
Informally, a GSF-local property $\mathcal{P}$ is non-propagating if repairing a graph $G$ that does not satisfy $\mathcal{P}$ does not trigger a global ``chain reaction'' of necessary modifications.  
We refer Section \ref{sec:gsf_preliminaries} for formal definitions. 

A major question that is left open is whether every GSF-local property satisfies the non-propagation condition. 

\subsection{Our contribution}

In this paper, we resolve the aforementioned open question raised in \cite{goldreich2011proximity}
by showing the following negative result.  
\begin{theorem}[Main result]\label{thm:existenceLocalNonTestableProperty}
	There exists a GSF-local property that is not testable in the bounded-degree graph model. Thus, not all GSF-local properties are non-propagating.  
\end{theorem}

We expect our result would shed some light on a full characterization of testable properties in the bounded degree model. Indeed, in the recent work by Ito, Khoury and Newman \cite{ito2019characterization}, the authors gave a characterization of testable \emph{monotone}
graph properties and testable \emph{hereditary}
graph properties with one-sided error in the bounded-degree graph model; and they asked the open question ``\emph{is every property that is defined by a set of forbidden configurations testable?}'' Since their definition of a property defined by a set of ``forbidden configuration'' is equivalent to a GSF-local property, 
our main result also gives a negative answer to their question. 

\subsection{Proof outline}
The starting point of our proof is a recent result of the authors that there exists a first-order (FO) property that is not testable in the bounded degree graph model \cite{testingFO}, where a property $\mathcal{P}$ is said to be an FO property if it can be expressed by an FO formula, \ie  a quantified formula whose variables represent graph vertices, with predicates for equality and adjacency. 
Intuitively, each structure in the property given in \cite{testingFO} is a hybridization of a sequence of expander graphs and a tree structure, where the expander graphs are recursively constructed by the zig-zag product introduced by Reingold et al.~\cite{ZigZagProductIntroduction}. Here each level of the tree structure forms one member of the recursive sequence of expander graphs. It was shown that this property is both an FO property and a family of expanders, and the latter implies it is not testable (see \eg \cite{fichtenberger2019every}). We refer to Section \ref{sec:proofMainThm} and \cite{testingFO} for a detailed description of the property. 

By Gaifman's locality theorem \cite{gaifman1982local}, it is known that FO  can only express local properties. 
Indeed, Hanf's Theorem \cite{Hanf1965} implies that we can understand this locality as prescribing upper and lower bounds for the occurrence of certain local neighbourhood (isomorphism) types. 

On the other hand, a GSF-local property as defined in \cite{goldreich2011proximity} refers to the freeness of some constant-size \emph{marked} graphs, where a mark graph $F$ specifies an induced subgraph and how it `interacts' with the rest of the graph (see Definition \ref{def:gsf}). Intuitively, such a property just specifies a condition that the local neighbourhoods of a graph $G$ should satisfy, i.e., certain types of local neighbourhoods cannot not occur in $G$, or equivalently, these types have $0$ occurrences. 

Building upon the above observations, we establish a formal connection between FO properties and GSF-local properties. We first encode the possible bounds on occurrences of local neighbourhood types into what we call \emph{neighbourhood profiles}, and characterise FO definable properties of bounded degree relational structures as finite unions of properties defined by neighbourhood profiles (Lemma~\ref{lemma:FO-neighbourhood}). We then show that every FO formula defined by a non-trivial finite union of properties which in turn is defined by a so-called \emph{$0$-profiles}, \ie the prescribed lower bounds are all $0$, is GSF-local (Theorem~\ref{thm:subsetOfFOIsLocal}). Given the fundamental roles of local properties in graph theory, graph limits \cite{LovaszBook2012}, we believe this new connection is of independent interest. 

For technical reasons, we make use of a property $\classStruc{P}_{\zigzag}$ of \emph{relational structures} that can be expressed by some FO formula while it is \emph{not} testable in the bounded-degree model, instead of directly using the non-testable graph property from \cite{testingFO}. We further prove that a minor variant of the {relational structure} property $\classStruc{P}_{\zigzag}$, which we denote by $\classStruc{P}_{\zigzag}'$, can be defined by $0$-profiles (Lemma \ref{lem:neighbouhoodProfilOfPZigZag}). Finally, we construct a non-testable \emph{graph} property $\graphProp$ by a \emph{local} reduction from the $\sigma$-structure property $\classStruc{P}_{\zigzag}'$ (Lemma~\ref{lem:local_reduction}). In the reduction we maintain being definable by $0$-profiles which proves GSF-locality of the graph property $\graphProp$ (Lemma~\ref{lemma:graphproperty_gsf_local}). Intuitively, the property $\graphProp$ encodes the property $\classStruc{P}_{\zigzag}$ in undirected graphs. Again, $\graphProp$ is a family of expanders (which guarantees non-testability), where in addition the local neighbourhoods satisfy the aforementioned features which guarantee that it is an FO property and also GSF-local.

\subsection{Other related work} The notion of POT was implicitly defined in \cite{blum1993self}. Goldreich and Shinkar \cite{goldreich2016two}  studied two-sided error POTs for both dense graph and bounded-degree graph models. Goldreich and Kaufman \cite{goldreichkauf2011proximity} investigated the relation between local conditions that are invariant in an adequate sense and properties that have a constant-query proximity-oblivious testers. 
Fichtenberger et al.~\cite{fichtenberger2019every} showed that every testable property is either finite or contains an infinite hyperfinite subproperty. 
\section{Preliminaries}\label{sec:preliminaries} 
\subsection{Graphs, relational structures and first-order logic}
We will briefly introduce structures and first-order logic and point the reader to~\cite{EF95} for a more detailed introduction.
A  (relational) \textit{signature} is a finite set $\sigma =\{R_1,\dots,R_\ell\}$ of relation symbols $R_i$. Every relation symbol $R_i$  has an arity  $\ar(R_i)\in \Npos$.
A \textit{$\sigma$-structure} is a tuple $\struc{A}=(\univ{A},\rel{R_1}{\struc{A}},\dots,\rel{R_\ell}{\struc{A}})$, where $\univ{A}$ is a \emph{finite} set, called the \emph{universe} of $\struc{A}$ and $\rel{R_i}{\struc{A}}\subseteq \univ{A}^{\ar(R_i)}$ is an $\ar(R_i)$-ary relation on $\univ{A}$. Note that if $\sigma=\{E_1,\ldots,E_{\ell}\}$ is a signature where each $E_i$ is a binary relation
symbol, then $\sigma$-structures are directed graphs 
with $\ell$ edge-colours. Let $\sigma_{\operatorname{graph}}:=\{E\}$ be a signature with one binary relation symbol $E$. Then we can understand undirected graphs as $\sigma_{\operatorname{graph}}$-structures for which the relation $E$ is symmetric (every undirected edge is represented by two tuples). Using this we can transfer all notions defined below for graphs. Typically we name graphs $G,H,F$, we denote the set of vertices of a graph $G$ by $V(G)$, the set of edges by $E(G)$ and vertices are typically named $u,v,w,u',v',w',\dots$. In contrast when we talk about a general relational structure we use $A,B$ and $a,b,a',b',\dots$ to denote elements from the universe.

In the following we let $\sigma$ be a relational signature.
 Two $\sigma$-structures $\struc{A}$ and $\struc{B}$ are \emph{isomorphic} if there is a bijective map from $\univ{A}$ to $\univ{B}$ that preserves all relations. 
For a $\sigma$-structure $\struc{A}$ and a subset $S\subseteq \univ{A}$, we let $ \struc{A}[S]$
denote the \emph{substructure} of $ \struc{A}$ \emph{induced} by $S$, i.\,e.\ $ \struc{A}[S]$ has universe $S$ and $\rel{R}{\struc{A}[S]}:=\rel{R}{ \struc{A}}\cap S^{\text{ar}(R)}$ for all $R\in \sigma$.
The \emph{degree} of an element $a\in \univ{A}$ denoted by $\deg_{\struc{A}}(a)$ is defined to be the number of tuples in $\struc{A}$ 
 containing $a$.
We define the \textit{degree} of $\struc{A}$, denoted by $\deg(\struc{A})$, to be the maximum degree of its elements. 
Given a signature $\sigma$ and a constant $d$, we let $\classStruc{C}_{\sigma,d}$ be the class of bounded-degree $d$ $\sigma$-structures and $\mathcal{C}_d$ the set of all bounded-degree $d$ graphs. Note that the degree of a graph differs by exactly a factor $2$ from the degree of the corresponding $\sigma_{\operatorname{graph}}$-structure.

Syntax and semantic of FO is defined in the usual way (see \eg \cite{EF95}). 
We use $\exists^{\geq m}x\,\phi$ (and $\exists^{= m}x\,\phi$, $\exists^{\leq m}x\,\phi$, respectively)
as a shortcut for the FO formula expressing that  the number of witnesses $x$ satisfying $\phi$
is at least $m$ (exactly $m$, at most $m$, respectively).
We say that a variable occurs \emph{freely} in an FO formula if at least one of its occurrences is not bound by any quantifier.
We use $\varphi(x_1,\dots,x_k)$ to express that the set of variables which occur  freely in the FO formula $\varphi$ is a subset of $\{x_1,\dots,x_k\}$. For a formula $\varphi(x_1,\dots,x_k)$, a $\sigma$-structure $\struc{A}$ and $a_1,\dots,a_k\in \univ{A}$ we write $\struc{A}\models \varphi(a_1,\dots,a_k)$ if $\varphi$ evaluates to true after assigning $a_i$ to $x_i$, for $1\leq i\leq k$. A \emph{sentence} of FO is a formula with no free variables. For an FO sentence $\varphi$ we say that $\struc{A}$ is a \emph{model} of $\varphi$ or $\struc{A}$ satisfies $\varphi$ if $\struc{A}\models \varphi$.

The \textit{Gaifman graph} of a $\sigma$-structure $\struc{A}$ is the undirected graph $\gaifman{\struc{A}}=(\univ{A},E)$, where $\{v,w\}\in E$, if $v\not=w$ and there is an $R\in \sigma$ and a tuple $\overline{a}=(a_1,\dots,a_{\ar(R)})\in \rel{R}{\struc{A}}$, such that $v=a_j$ and $w=a_k$ for some $1\leq k,j\leq \ar(R)$. We use $\gaifman{\struc{A}}$ to apply graph theoretic notions to relational structures. Note that for any graph the Gaifman graph of the corresponding  symmetric $\sigma_{\operatorname{graph}}$-structure is the graph itself.
For two elements $a,b\in \univ{A}$, we define the \emph{distance} between $a$ and $b$ in $\struc{A}$, denoted by
$\dist_{\struc{A}}(a,b)$, as the length of a shortest path form $a$ to $b$ in $\gaifman{\struc{A}}$, or $\infty$ if there is no such path. 
For $r\in \mathbb{N}$ and   $a\in \univ{A}$, the \textit{$r$-neighbourhood} of $a$ is the set $N_r^{\struc{A}}(a):=\{b\in \univ{A}: \dist_{\struc{A}}(a,b)\leq r\}$. We define $\mathcal{N}_r^{\struc{A}}(a):=\struc{A}[N_r^{\struc{A}}(a)]$ to be the substructure of $\struc{A}$ induced by the $r$-neighbourhood of $a$. 
For $r\in \mathbb{N}$ an \emph{$r$-ball} is a tuple $(\struc{B},b)$, where $\struc{B}$ is a $\sigma$-structure, $b\in \univ{B}$ and $\univ{B}=N_r^{\struc{B}}(b)$, \ie $\struc{B}$ has radius $r$ and $b$ is the centre.
Note that by definition $(\mathcal{N}_r^{\struc{A}}(a),a)$ is an $r$-ball for any $\sigma$-structure $\struc{A}$ and $a\in \univ{A}$. Two $r$-balls $(\struc{B},b),(\struc{B}',b')$ are isomorphic if there is an isomorphism of $\sigma$-structure from $\struc{B}$ to $\struc{B}'$ that maps $b$ to $b'$.
We call the isomorphism classes of $r$-balls \emph{$r$-types}. 
For an $r$-type $\tau$ and an element $a\in \univ{A}$ we say that $a$ \emph{has} ($r$-)type $\tau$  if $(\mathcal{N}_r^{\struc{A}}(a),a)\in \tau$. 
Moreover, given such an $r$-type $\tau$, there is a formula $\phi_{\tau}(x)$ such that 
for every $\sigma$-structure $\struc{A}$ and for every $a\in \univ{A}$, $\struc{A}\models\phi_{\tau}(a)$ iff
$(\mathcal{N}_r^{\struc{A}}(a),a)\in \tau$.
A \emph{Hanf-sentence} is a sentence of the form $\exists ^{\geq m} x \phi_{\tau}(x)$, for some $m\in\Npos$, where $\tau$ is an 
$r$-type. 
An FO sentence is in \emph{Hanf normal form}, if it is a Boolean combination\footnote{By Boolean combination we 
	always mean \emph{finite} Boolean combination.} of Hanf sentences.
Two formulas $\phi(x_1,\dots,x_k)$ and $\psi(x_1,\dots,x_k)$ of signature $\sigma$ are called
\emph{$d$-equivalent}, if they are equivalent on $\classStruc{C}_{\sigma,d}$, i.\,e.\ for all $\struc{A}\in \classStruc{C}_{\sigma,d}$ and
$(a_1,\dots,a_k) \in \univ{A}^{k}$  we have 
$\struc{A}\models\phi( a,\dots,a_k)$ iff $\struc{A}\models\psi(a_1,\dots,a_k)$.
Hanf's locality theorem for first-order logic~\cite{Hanf1965} implies the following.

\begin{theorem}[Hanf~\cite{Hanf1965}]\label{thm:Hanf}
	Let $d\in\mathbb N$. Every sentence of first-order logic is $d$-equivalent to a 
	sentence in Hanf normal form.
\end{theorem}
\subsection{Property testing}
In the following, we give definitions of two models for property testing - the bounded-degree model for graphs and the bounded-degree model for relational structures. For notational convenience, $\mathcal{C}$ will either denote a class of graphs of bounded-degree $d$, or a class of $\sigma$-structures of bounded-degree $d$ for some signature $\sigma$ and some $d\in \mathbb{N}$. We will further refer to both graphs and $\sigma$-structures as structures. A \emph{property} $\mathcal{P}$ in $\mathcal{C}$ is a subset of $\mathcal{C}$ which is closed under isomorphism. We say that a structure $\struc{A}$ has property $\mathcal{P}$ if $\struc{A}\in \mathcal{P}$. For $\epsilon\in (0,1)$ we say that a structure $\mathcal{A}$ on $n$ vertices/elements is \emph{$\epsilon$-close} to $\mathcal{P}$  if there is a structure $\struc{A}'\in \mathcal{P}$ such that $\struc{A}$ and $\struc{A}'$ differ in at most $\epsilon d n$ edges/tuples. We say that $\struc{A}\in \mathcal{C}$ is $\epsilon$-far from $\mathcal{P}$ if $\struc{A}$ is not $\epsilon$-close to $\mathcal{P}$.

A property tester accesses a structure via oracle queries. A \emph{query} to a $\sigma$-structure  $\struc{A}$ of bounded-degree $d$ has the form $(a,i)$ for an element $a\in \univ{A}$, $i\in \{1,\dots,d\}$ and is answered by $\operatorname{ans}(a,i):=(R,a_1,\dots,a_{\ar(R)})$ where $(a_1,\dots,a_{\ar(R)})$ is the $i$-th tuple containing $a$ and $(a_1,\dots,a_{\ar(R)})\in \rel{R}{\struc{A}}$.  A \emph{query} to a graph $G$ of bounded-degree $d$ has the form $(v,i)$ for  $v\in V(G)$, $i\in \{1,\dots,d\}$ and is answered by $\operatorname{ans}(v,i):=w$ where $w$ is the $i$-th neighbour of $v$.   

Let $\mathcal{P}_n$ be the subset of $\mathcal{P}$ with $n$ vertices/elements. Thus $\mathcal{P}=\cup_{n\in \mathbb{N}}\mathcal{P}_n$. We give the formal definitions of standard property testing and proximity-oblivious testing in Appendix \ref{app:preliminaries}.  

\subsection{Generalised subgraph freeness}\label{sec:gsf_preliminaries}
Now we present the formal definition of generalised subgraph freeness, GSF-local properties and the notion of non-propagation, which were introduced in \cite{goldreich2011proximity}.  
\begin{definition}[Generalized subgraph freeness (GSF)
	]\label{def:gsf}
	A \emph{marked} graph is a graph with each vertex marked as either \emph{`full'} or \emph{`semifull'} or \emph{`partial'}. An \emph{embedding} of a marked graph $F$ into a graph $G$ is an injective map $f:V(F)\rightarrow V(G)$ such that for every $v\in V(F)$ the following three conditions hold.
	\begin{enumerate}
		\item If $v$ is marked `full', then  
		$N_1^G(f(v))=f(N_1^F(v))$.
		\item If $v$ is marked `semifull', then 
		$N_1^G(f(v))\cap f(V(F))=f(N_1^F(v))$.
		\item If $v$ is marked `partial', then 
		$N_1^G(f(v))\supseteq f(N_1^F(v))$.
	\end{enumerate}
	The graph $G$ is called $F$-free if  there is no embedding of $F$ into $G$. For a set of marked graphs $\mathcal{F}$, a graph $G$ is called $\mathcal{F}$-free if it is $F$-free for every $F\in \mathcal{F}$. 
\end{definition}
Based on the above definition of GSF, we can define GSF-local properties. 
\begin{definition}[GSF-local properties]
	Let $\mathcal{P}=\cup_{n\in \mathbb{N}}\mathcal{P}_n$ be a graph property where $\mathcal{P}_n=\{G\in \mathcal{P}\mid |V(G)|=n\}$ and  $\overline{\mathcal{F}} = (\mathcal{F}_n)_{n \in \mathbb{N}}$ a sequence of sets of
	marked graphs. $\mathcal{P}$ is called \emph{$\overline{\mathcal{F}}$-local} if there exists an integer $s$ such that for every $n$ the following conditions hold.
	\begin{enumerate}
		\item $\mathcal{F}_n$ is a set of marked graphs, each of size at most $s$.
		\item $\mathcal{P}_n$ equals the set of $n$-vertex graphs that are $\mathcal{F}_n$-free. 
	\end{enumerate}
	$\mathcal{P}$ is called \emph{GSF-local} if there is a sequence $\overline{\mathcal{F}} = (\mathcal{F}_n)_{n \in \mathbb{N}}$  of sets of
	marked graphs such that $\mathcal{P}$ is $\overline{\mathcal{F}}$-local.
\end{definition}
The following notion of non-propagating condition of a sequence of sets of marked graphs was introduced to study constant-query POTs.
\begin{definition}[Non-propagating]
	Let $\overline{\mathcal{F}} = (\mathcal{F}_n)_{n \in \mathbb{N}}$ be a sequence of sets of
	marked graphs.
	\begin{itemize}
		\item For a graph $G $, a subset $B \subset V(G)$ \emph{covers} $\mathcal{F}_n$ in $G$ if for every marked
		graph $F \in  \mathcal{F}_n$ and every embedding of $F$ in $G$, at least one vertex of $F$ is mapped to a vertex
		in $B$.
		\item The sequence $\overline{\mathcal{F}}$ is \emph{non-propagating} if there exists a (monotonically non-decreasing) function
		$\tau: (0, 1] \rightarrow (0, 1]$ such that the following two conditions hold.
		\begin{enumerate}
			\item For every $\epsilon > 0$ there exists $\beta > 0$ such that $\tau(\beta) < \epsilon$.
			\item For every graph $G$ and every $B \subset V(G)$ such that $B$ covers $\mathcal{F}_n$ in $G$, either $G$
			is $\tau(|B|/n )$-close to being $\mathcal{F}_n$-free or there are no $n$-vertex graphs that are $\mathcal{F}_n$-free. 
		\end{enumerate}
			A GSF-local property $\mathcal{P}$ is \emph{non-propagating} if there exists a non-propagating sequence $\overline{\mathcal{F}}$ such
		that $\mathcal{P}$ is $\overline{\mathcal{F}}$-local.
	\end{itemize}
\end{definition}
In the above definition, the set $B$ can be viewed as the set involving necessary modifications for repairing a graph $G$ that does not satisfy the property $\mathcal{P}$ that is $\overline{\mathcal{F}}$-local, and the second condition says we do not need to modify $G$ ``much beyond'' $B$. In particular, it implies we can repair 
$G$ without triggering a global ``chain reaction''. 
Goldreich and Ron gave the following characterization for the proximity-oblivious testable properties in the bounded-degree graph model. 
\begin{theorem}[Theorem 5.5 in \cite{goldreich2011proximity}]\label{thm:charOfPOT}
	A graph property $\mathcal{P}$ has a
	constant-query proximity-oblivious tester if and only if $\mathcal{P}$ is GSF-local and non-propagating.
\end{theorem}

The following open question was raised in \cite{goldreich2011proximity}. 
\begin{question}[Are all GSF-local properties non-propagating?]
Is it the case that for every GSF-local property $\mathcal{P}=\cup_{n\in \mathbb{N}}\mathcal{P}_n$,  there is a sequence $\overline{\mathcal{F}} = (\mathcal{F}_n)_{n \in \mathbb{N}}$ that is non-propagating and $\mathcal{P}$ is $\overline{\mathcal{F}}$-local?
\end{question}
\section{Relating different notions of locality}\label{sec:relatingNotionsOfLocality}
In this section we define properties by prescribing upper and lower bounds on the number of occurrence of  neighbourhood types. These bounds are given by \emph{neighbourhood profiles} which we will define formally below. 
We use these properties to give a natural characterization of FO properties of bounded-degree structures in Lemma~\ref{lemma:FO-neighbourhood}, which is a straightforward consequence of Hanf's Theorem (Theorem~\ref{thm:Hanf}). We use this characterization to establish links between FO definability and GSF-locality. This connection is  the key ingredient in the proof of our main theorem.\\

Observe that for fixed $r,d\in \mathbb N$ and $\sigma$, there are only finitely many $r$-types in structures in $\classStruc{C}_{\sigma,d}$.
For any signature $\sigma$ and  $d,r\in \mathbb{N}$ we let $n_{d,r,\sigma}\in \mathbb{N}$ be the number of different $r$-types of $\sigma$-structures of degree at most $d$. Assuming that for all $d,r\in \mathbb{N}$ the $r$-neighbourhood-types of $\sigma$-structures of degree at most $d$ are ordered, we let $\tau_{d,r,\sigma}^i$ denote the $i$-th 
such neighbourhood type, for $i\in \{1,\dots,n_{d,r,\sigma}\}$. 
With each $\sigma$-structure $\struc{A}\in \classStruc{C}_{\sigma,d}$ we associate its 
\emph{$r$-histogram vector} $\vet{v}_{d,r,\sigma}(\struc{A})$, given by 
\begin{displaymath}
	(\vet{v}_{d,r,\sigma}(\struc{A}))_i:=|\{a\in \univ{A}\mid \mathcal{N}_{r}^{\struc{A}}(a)\in \tau_{d,r,\sigma}^i\}|.
\end{displaymath} 
We let 
\begin{displaymath}
	\mathfrak{I}:=\{[k,l],[k,\infty)\mid k\leq l\in \mathbb{N}\}
\end{displaymath}
be the set of all closed or half-closed, infinite intervals with natural lower/upper bounds.
\begin{definition}
	Let $\sigma$ be a signature and $d,r\in \mathbb{N}$.
	\begin{enumerate}
		\item An \emph{$r$-neighbourhood profile} 
		of degree $d$ is a function $\rho:\{1,\dots,n_{d,r,\sigma}\}\rightarrow \mathfrak{I}$. 
		\item For a structure $\struc{A}\in \classStruc{C}_{\sigma,d}$, we say $\struc{A}$ obeys $\rho$, denoted by $\struc{A}\sim \rho$, if
		\[
		(\vet{v}_{d,r,\sigma}(\struc{A}))_i\in \rho(i) \text{ for all }i\in \{1,\dots,n_{d,r,\sigma}\}.
		\]
		Let $\classStruc{P}_\rho$ be the set of structures $\struc{A}$ that obey $\rho$, i.e., $\classStruc{P}_\rho=\{\struc{A}\in \classStruc{C}_{\sigma,d}\mid \struc{A}\sim \rho\}$.
	\item We say that a property $\classStruc{P}$ is \emph{defined by a finite union of neighbourhood profiles} if there is $k\in \mathbb{N}$ such that  $\classStruc{P}=\bigcup_{1\leq i \leq k}\classStruc{P}_{\rho_i}$ where $\rho_i$ is an $r_i$-neighbourhood profile and $r_i\in \mathbb{N}$ for every $i\in \{1,\dots,k\}$. 
	\end{enumerate}
\end{definition}

We let $n_{d,r}:=n_{d,r,\sigma_{\operatorname{graph}}}$ denote the total number of $r$-type of 
undirected graphs of degree at most $d$, and let $\tau_{d,r}^i:=\tau_{d,r,\sigma_{\operatorname{graph}}}^i$ be the $i$-th $r$-type of bounded degree $d$, for any $i\in \{1,\dots,n_{d,r}\}$. Further, for a graph $G$ let $\vet{v}_{d,r}(G)$ denote the $r$-histogram vector of $G$.  Note that for any type $\tau_{d,r}^i$ where the edge relation is not 
symmetric we have that $(\vet{v}_{d,r}(G))_i=0$ and therefore 
in any $r$-neighbourhood profile $\rho$ for graphs we have $\rho(i)=[0,0]$ for any type $\tau_{d,r}^i$ which is not symmetric.   

We now give a lemma showing that bounded-degree FO properties can be equivalently defined as finite unions of properties defined by neighbourhood profiles. Here the technicalities  that arise are due to Hanf normal form not requiring the locality-radius  of all Hanf-sentences to be the same. The proof of Lemma~\ref{lemma:FO-neighbourhood} is deferred to Appendix~\ref{app:C}.

\begin{lemma}\label{lemma:FO-neighbourhood}
	For every non-empty property $\classStruc{P}\subseteq \classStruc{C}_{\sigma,d}$, $\classStruc{P}$ is FO definable on $\classStruc{C}_{\sigma,d}$ if and only if
	$\classStruc{P}$ can be obtained as a finite union of properties defined by neighbourhood profiles.
\end{lemma}

\subsection{Relating FO properties to GSF-local properties}   
We now prove that FO properties which arise as unions of neighbourhood profiles of a particularly simple form are GSF-local. 
For this let 
\begin{displaymath}
\mathfrak{I}_{0}:=\{[0,\infty),[0,k]\mid k\in \mathbb{N} \}\subset \mathfrak{I}. 
\end{displaymath}
We call any neighbourhood profile $\rho$ with codomain $\mathfrak{I}_{0}$ 
 a \emph{$0$-profile}, as all lower bounds for the occurrence of types are $0$.
 \begin{observation}\label{obs:expressingExOfType}
	 Let $\rho$ be a $0$-profile. If two structures $\struc{A},\struc{A}'\in \classStruc{C}_{\sigma,d}$ satisfy $(\vet{v}_{d,r,\sigma}(\struc{A}))_i\leq (\vet{v}_{d,r,\sigma}(\struc{A}'))_i$ for every $i\in\{1,\dots,n_{d,r,\sigma}\}$ and $\struc{A}'\sim\rho$, then $\struc{A}\sim \rho$. \\
	 In particular, the existence of an $r$-type cannot be expressed by a $0$-profile. 
 \end{observation}

\begin{theorem}\label{thm:subsetOfFOIsLocal}
	Every finite union of  properties defined by $0$-profiles is GSF-local.
\end{theorem} 
\begin{proof}
We prove this in two parts (Claim~\ref{claim:GSFProfileIsGSFLocal} and Claim~\ref{claim:GSFLocalClosedUnderUnion}). We first argue that every property $\classStruc{P}_{\rho}$  
defined by some $0$-profile $\rho:\{1,\dots,n_{d,r,\sigma}\}\rightarrow\mathfrak{I}_{0}$ 
is GSF-local. For this it is important to note that we can express a forbidden $r$-type $\tau$ by a forbidden generalised subgraph. 
For $(B,b)\in \tau$, the set of all graphs with no vertex of neighbourhood type $\tau$ is the set of all $B$-free graphs where every vertex in $V(B)$ of distance less than $r$ to $b$ is marked `full' and every vertex in $V(B)$ of distance $r$ to $b$ is marked `semifull'. 
	Since a profile of the form $\rho:\{1,\dots,n_{d,r,\sigma}\}\rightarrow\mathfrak{I}_{0}$ can express that some neighbourhood type $\tau$ can appear at most $k$ times for some fixed $k\in \mathbb{N}$, we need to forbid all marked graphs in which type $\tau$ appears $k+1$ times. We will formalise this in the following claim. 
\begin{claim}\label{claim:GSFProfileIsGSFLocal}
	For every $r$-neighbourhood profile $\rho: \{1,\dots,n_{d,r}\}\rightarrow \mathfrak{I}_{0}$, there is a finite set $\cal{F}$ of marked graphs such that $\classStruc{P}_\rho$ is exactly the property of $\mathcal{F}$-free graphs. 
\end{claim}
\begin{proof}
	Assume $\tau$ is an $r$-type and $k\in \mathbb{N}_{>0}$. Then we say that a marked graph $F$ is a \emph{$k$-realisation} of $\tau$ if $F$ has the following properties.
	\begin{enumerate}
		\item There are $k$ distinct vertices $v_1,\dots,v_k$ in $F$ such that $(\mathcal{N}_r^F(v_i),v_i)\in \tau$ for every $i=1,\dots,k$.
		\item Every vertex $v$ in $F$ has distance less  or equal to $r$ to at least one vertex $v_i$.
		\item Every vertex $v$ in $F$ of distance less than $r$ to at least one $v_i$ is marked as `full'.
		\item Every vertex $v$ in $F$ of distance greater or equal to $r$ to every $v_i$ is marked as `semifull'.
	\end{enumerate}  
	We denote by $S^k(\tau)$ the set of all $k$-realisations of $\tau$.
	
	Now we can define the set $\mathcal{F}$  of forbidden subgraphs to be
	\[
	{\cal F}:=\bigcup_{k\in \mathbb{N}, 1\leq i\leq n_{d,r,\sigma}: \rho(i)=[0,k]} 
	S^{k+1}(\tau_{d,r}^i).
	\]
	
	Let $\mathcal{P}$ be the property of all $\mathcal{F}$-free graphs. We first prove that the property $\mathcal{P}$ is contained in  $\classStruc{P}_\rho$. Towards a contradiction  assume that $G\in \mathcal{C}_d$ is ${\cal F}$-free  but not contained in $\classStruc{P}_{\rho}$. As $G$ is not contained in $\classStruc{P}_{\rho}$ there must be an index $i\in \{1,\dots,n_{d,r}\}$ such that $(\vet{v}_{d,r}(G))_i\notin \rho(i)$. Since $\rho(i) \in \mathfrak{I}_{0}$ there is $k\in \mathbb{N}$ such  that $\rho(i)=[0,k]$ and hence $(\vet{v}_{d,r}(G))_i> k$. Hence there must be $k+1$ vertices $v_1,\dots,v_{k+1}$ in $G$ such that $(\mathcal{N}_r^G(v_i),v_i)\in \tau_{d,r}^i$. 
	We define the  marked graph  $F$ to be the subgraph of $G$ induced by the $r$-neighbourhoods 
	of $v_1,\dots,v_{k+1}$, \ie $G[\cup_{1\leq i\leq k+1}N_r^G(v_i)]$, in which every vertex of distance less than $k$ to at least one of the $v_i$ is marked as `full' and every other vertex is marked as `semifull'.  
	Then $F$ is by definition a $(k+1)$-realisation of $\tau_{d,r}^i$ and hence $F\in {\cal F}$.
	We now argue that $F$ can be embedded into $G$. Since $F$ is an induced subgraph of $G$ the identity map gives us a natural embedding $f:F\rightarrow G$. Let $v$ be any vertex marked `full' in $F$. Then by construction of $F$, there is $i\in \{1,\dots,k+1\}$ such that $f(v)$ is of distance less than $r$ to $v_i$ in $G$. But then $N_1^G(f(v))$ is a subset of 
	$N_r^G(v_i)$. As $F$ without the marking is the subgraph of $G$ 
	induced by $\cup_{1\leq i\leq k+1}N_r^G(v_i)$ this implies that $f(N_1^F(v))=N_1^G(f(v))$.  Furthermore, assume $v$ is a vertex marked `semifull' in $F$. Then $f(N_1^F(v))= N_1^G(f(v))\cap f(V(F))$ holds as $F$ without the markings is an induced subgraph of $G$.
	This proves that $G$ is not $F$-free by Definition~\ref{def:gsf}. This is a contradiction to our assumption that $G$ is $\mathcal{F}$-free and $F\in \mathcal{F}$. 
		
	Similarly, we can show that $\classStruc{P}_\rho\subseteq \mathcal{P}$ by assuming $G\in \mathcal{C}_d$ is in $\classStruc{P}_\rho$  but not ${\cal F}$-free, and showing that the embedding of any graph of $\mathcal{F}$ into $G$ yields an amount of vertices of a certain type contradicting containment in $\classStruc{P}_\rho$. 
\end{proof}
Next we prove that 
classes defined by excluding finitely many marked graphs are closed under finite unions.	
\begin{claim}\label{claim:GSFLocalClosedUnderUnion}
	Let $\mathcal{F}_1, \mathcal{F}_2$ be two finite sets of marked graphs. 
	For $i\in\{1,2\}$, let $\mathcal{P}_i$ be the property of $\mathcal{F}_i$-free graphs. 
	Then there is a set $\mathcal{F}$ of generalised subgraphs such that $\mathcal{P}_1\cup\mathcal{P}_2$ is the property of $\mathcal{F}$-free graphs. 
\end{claim}
\begin{proof}
	We say that a marked graph $F$ is a (not necessarily disjoint) union of  marked graphs $F_1, F_2$ if 
	\begin{enumerate}
		\item  there is an embedding $f_i$ of $F_i$ into the graph $F$ without its markings as in Definition~\ref{def:gsf} for every $i\in \{1,2\}$.
		\item for every vertex $v$ in $F$ there is $i\in\{1,2\}$ and a vertex $w$ in $F_i$ such that $f_i(w)=v$.
		\item every vertex $v$ in $F$ is marked `full', if there is  $i\in\{1,2\}$ and a `full' vertex $w$ in  $F_i$ such that $f_i(w)=v$.
		\item every vertex $v$ in $F$ is marked `semifull', if there is  $i\in \{1,2\}$ and a `semifull' vertex $w$ in  $F_i$ such that $f_i(w)=v$ and $f_i(u)\not=v$ for every $i\in\{1,2\}$ and every `full' vertex $u$.
		\item every vertex $v$ in $F$ is marked `partial' if $f_i(u)\not=v$ for every $i\in\{1,2\}$ and every `full' or `semifull' vertex $u$.
	\end{enumerate}  
	We define $S(F_1,F_2)$ to be the set of all possible (not necessarily disjoint) unions of $F_1,F_2$. 
	We can now define the set $\mathcal{F}$ to be 
	\[\mathcal{F}:=\bigcup_{F_1\in \mathcal{F}_1,F_2\in \mathcal{F}_2} S(F_1,F_2).\] 
	
	Let $\mathcal{P}$ be the property of all $\mathcal{F}$-free graphs.
	Now we prove $\mathcal{P}\subseteq \mathcal{P}_1\cup\mathcal{P}_2$. Towards a contradiction assume $G$ is $\mathcal{F}$-free but $G$ is in neither $\mathcal{P}_1$ nor in $\mathcal{P}_2$.  
	Then for every $i\in \{1,2\}$ there is a graph $F_i\in \mathcal{F}_i$ such that $G$ is not $F_i$-free. It is easy to see that  there is a  union $F_\cup$ of $F_1$ and $F_2$    
	such that $G$ is not $F_\cup$-free, which contradicts that $G$ is $\mathcal{F}$-free. 
	
	Conversely, in order to prove $\mathcal{P}_1\cup \mathcal{P}_2 \subseteq \mathcal{P}$, if $G$ is $\mathcal{F}_i$ free for some 
	$i\in\{1,2\}$ then $G$ must be $\mathcal{F}$-free by construction of $\mathcal{F}$.
\end{proof}
Combining the two claims above proves the  Theorem~\ref{thm:subsetOfFOIsLocal}.
\end{proof}
	\paragraph{Further discussion of the relation between FO and GSF-locality} First let us remark that it 
	is neither true that every FO definable property
	is GSF-local, nor that every GSF-local property is FO definable. 
	\begin{example}\label{exa:FONotContainedInGSF}
		The property of bounded-degree  graphs containing a triangle is FO definable but not GSF-local.
	\end{example}
	Indeed, the existence  of 
	a fixed number of vertices of certain neighbourhood types can be expressed in FO, while in general, this cannot be expressed by forbidding generalised subgraphs. 
	If a formula has a $0$-profile 
	(and hence does not require the existence of any types) 
	then the property defined by that formula is GSF-local, as shown in Theorem~\ref{thm:subsetOfFOIsLocal}.
	\begin{example}\label{exa:GSFNotContainedInFO}
		The class of all bounded-degree graphs with an even number of vertices is GSF-local but not FO definable.
	\end{example}
	Let us remark that Theorem~\ref{thm:subsetOfFOIsLocal} combined with Lemma~\ref{lemma:FO-neighbourhood} proves that every finite union of properties definable by $0$-profiles is both FO definable and GSF-local. 
	Hence it is natural to ask whether the intersection of FO definable properties and GSF-local properties is precisely the set of finite unions of properties definable by $0$-profiles. However, this is not the case.
	The following example shows that there are properties which are both FO definable and GSF-local but cannot be expressed by $0$-profiles.
	\begin{figure*}
		\centering
		\begin{tikzpicture}
			\tikzstyle{ns1}=[line width=1.2]
			\tikzstyle{ns2}=[line width=1.2]
			\definecolor{C1}{RGB}{1,1,1}
			\definecolor{C2}{RGB}{0,0,170}
			\definecolor{C3}{RGB}{251,86,4}
			\definecolor{C4}{RGB}{50,180,110}
			\def \dist {0.7}
			\def \heightWriting {0.4}
			\def \heightName {-1}
			\node[circle,fill=white,inner sep=0pt, minimum width=0pt] (15) at (0,3) {};
			\node[draw,circle,fill=black,inner sep=0pt, minimum width=5pt] (0) at (0,0) {};
			\node[draw,circle,fill=black,inner sep=0pt, minimum width=5pt] (5) at (3,0) {};
			\node[draw,circle,fill=black,inner sep=0pt, minimum width=5pt] (6) at (4,0) {};
			\node[draw,circle,fill=black,inner sep=0pt, minimum width=5pt] (7) at (5,0) {};
			\node[draw,circle,fill=black,inner sep=0pt, minimum width=5pt] (9) at (8,0) {};
			\node[draw,circle,fill=black,inner sep=0pt, minimum width=5pt] (10) at (9,0) {};
			
			\draw[ns1] (5)--(6)--(7);

			\node[minimum height=10pt,inner sep=0,font=\scriptsize] at (0,\heightWriting) {partial};
			\node[minimum height=10pt,inner sep=0,font=\scriptsize] at (3,\heightWriting) {partial};
			\node[minimum height=10pt,inner sep=0,font=\scriptsize] at (4,-\heightWriting) {partial};
				\node[minimum height=10pt,inner sep=0,font=\scriptsize] at (5,\heightWriting) {partial};
			\node[minimum height=10pt,inner sep=0,font=\scriptsize] at (8,\heightWriting) {full};
			\node[minimum height=10pt,inner sep=0,font=\scriptsize] at (9,\heightWriting) {full};
			\node[minimum height=10pt,inner sep=0] at (0,\heightName) {$G_1$};
			\node[minimum height=10pt,inner sep=0] at (4,\heightName) {$G_2$};
			\node[minimum height=10pt,inner sep=0] at (8.5,\heightName) {$G_3$};

			\end{tikzpicture} 
		\caption{Marked graphs for Example~\ref{ex:0ProfilesNotEntireIntersecion}.}\label{fig:setOfMarkedGraphs}
		
	\end{figure*}
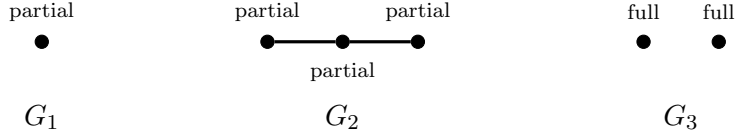
	\begin{example}\label{ex:0ProfilesNotEntireIntersecion} We let $d\geq 2$ and let  $B_1:=(\{v\},\{\})$, $B_2=(\{v,w\},\{\{v,w\}\})$ be two graphs. We further let $\tau_1,\tau_2$ be the $1$-types of degree $d$ such that $(B_1,v)\in \tau_1$ and $(B_2,v)\in \tau_2$. Consider the property $\mathcal{P}$ defined by the following FO formula
		\begin{displaymath}
			\varphi:=\lnot \exists x (x=x)\lor \exists^{=1}x\big(\varphi_{\tau_1}(x)\land \forall y (x\not=y\rightarrow \varphi_{\tau_2}(y))\big). 	
		\end{displaymath}
		$\mathcal{P}$ contains, besides the empty graph, unions of an arbitrary amount  of disjoint edges and one isolated vertex. To define a sequence of forbidden subgraphs we let $G_1,G_2,G_3$ be the marked graphs in Figure~\ref{fig:setOfMarkedGraphs}. Let $\mathcal{F}_{\operatorname{even}}:=\{G_1\}$ and $\mathcal{F}_{\operatorname{odd}}:=\{G_2,G_3\}$ and let $\overline{\mathcal{F}}=(\mathcal{F}_n)_{n \in \mathbb{N}}$ where $\mathcal{F}_i=\mathcal{F}_{\operatorname{even}}$ if $i$ is even and $\mathcal{F}_i=\mathcal{F}_{\operatorname{odd}}$ if $i$ is odd. Note that every graph on more than one 
		vertex with an odd number of vertices which is $\mathcal{F}_{\operatorname{odd}}$-free must contain a vertex of neighbourhood type $\tau_1$, 
		and that the set of $\mathcal{F}_{\operatorname{even}}$-free graphs contains only the empty graph. Hence $\mathcal{P}$ is $\overline{\mathcal{F}}$-local. Now assume towards a contradiction that $\mathcal{P}=\cup_{1\leq i\leq k}\mathcal{P}_{\rho_i}$ for $0$-profiles $\rho_i$. Let $G_m$ be the graph consisting of $m$ disjoint edges and one isolated vertex and $H_m$ the graph consisting of $m$ disjoint edges. Since $G_m\in\mathcal{P}$ there is $i\in \{1,\dots,k\}$ such that $G_m\sim \rho_i$. By choice of $G_m$ and $H_m$ we have $0\leq (\vet{v}_{d,r}(H_m))_j\leq (\vet{v}_{d,r}(G_m))_j\in \rho_i(j)$ for every $j\in \{1,\dots,n_{d,r}\}$. Since additionally $\rho_i(j)\in \mathfrak{I}_0$ this implies that $(\vet{v}_{d,r}(H_m))_j\in \rho_i(j)$. But then $H_m\sim \rho_i$ which yields a contradiction as $H_m\notin \mathcal{P}$. Hence $\mathcal{P}$ can not be defined as a finite union of $0$-profiles. 
	\end{example}
 Figure~\ref{fig:overview} gives a schematic overview of all classes of properties discussed here and their relationship.
	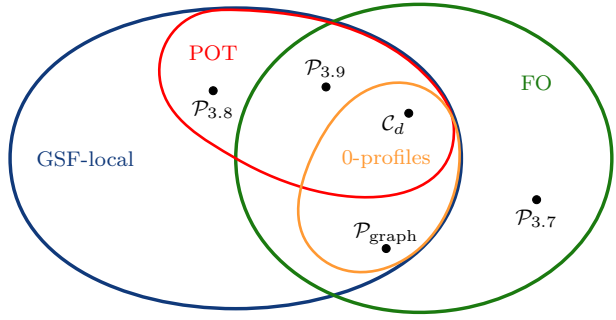
\begin{figure}
		\centering
		\begin{tikzpicture}[scale=1,use Hobby shortcut,closed=true]
		\tikzstyle{LW1}=[line width=1.3]
		\tikzstyle{LW3}=[line width=1]
		\definecolor{C4}{RGB}{255,153,51}
		\definecolor{C3}{RGB}{255,0,0}
		\definecolor{C2}{RGB}{27, 122, 16}
		\definecolor{C1}{RGB}{16, 57, 122}

		\draw[LW1,C1] ([closed]0,0)..(3,2)..(6,0)..(3,-2);
		\draw[LW1,C2] ([closed]3,0)..(6,2)..(8,0)..(6,-2);
		\draw[LW3,C3] ([closed]2,1)..(2.5,1.92)..(3,1.975)..(5.62,1)..(5.9,0.4)..(3,0);
		\draw[LW3,C4] ([closed]4,0)..(5.05,1)..(5.7,0.8)..(5.97,-0.1)..(4,-1.2);

		\node[C1,minimum height=10pt,inner sep=0,font=\scriptsize] at (1,0) {GSF-local};
		\node[C2,minimum height=10pt,inner sep=0,font=\scriptsize] at (7,1) {FO};
		\node[C3,minimum height=10pt,inner sep=0,font=\scriptsize] at (2.7,1.4) {POT};
		\node[C4,minimum height=10pt,inner sep=0,font=\scriptsize] at (5,0) {$0$-profiles};
		\node[draw,circle,fill=black,inner sep=0pt, minimum width=3pt] (1) at (5,-1.2) {};
		\node[minimum height=10pt,inner sep=0,font=\scriptsize] at (5,-1) {$\graphProp$};
		\node[draw,circle,fill=black,inner sep=0pt, minimum width=3pt] (1) at (4.2,0.95) {};
		\node[minimum height=10pt,inner sep=0,font=\scriptsize] at (4.2,1.2) {$\mathcal{P}_{\ref{ex:0ProfilesNotEntireIntersecion}}$};
		\node[draw,circle,fill=black,inner sep=0pt, minimum width=3pt] (1) at (2.7,0.9) {};
		\node[minimum height=10pt,inner sep=0,font=\scriptsize] at (2.7,0.65) {$\mathcal{P}_{\ref{exa:GSFNotContainedInFO}}$};
		\node[draw,circle,fill=black,inner sep=0pt, minimum width=3pt] (1) at (7,-0.55) {};
		\node[minimum height=10pt,inner sep=0,font=\scriptsize] at (7,-0.8) {$\mathcal{P}_{\ref{exa:FONotContainedInGSF}}$};
		\node[draw,circle,fill=black,inner sep=0pt, minimum width=3pt] (1) at (5.3,0.6) {};
		\node[minimum height=10pt,inner sep=0,font=\scriptsize] at (5.1,0.45) {$\mathcal{C}_d$};

		\end{tikzpicture}
		\caption{Overview of the classes of properties, here $\mathcal{P}_i$ refers to the property from Example $i$, $\mathcal{C}_d$ refers to the property of all graphs of bounded degree $d$ and $\graphProp$ is the property defined in Section~\ref{sec:localreduction}.}\label{fig:overview}
	\end{figure}

\section{Proof of the main theorem}\label{sec:proofMainThm}

In this section we prove Theorem~\ref{thm:existenceLocalNonTestableProperty}. We start by describing a property of relational structures, similar to a property 
in~\cite{testingFO}, which is not testable. We then show that the property 
can be expressed by a union of $0$-profiles, and hence by Theorem~\ref{thm:subsetOfFOIsLocal} it is GSF-local. 

Let $\sigma$ be the signature, $d\in \mathbb{N}$  and $\classStruc{P}_{\zigzag}$ be the property of $d$ $\sigma$-structures of bounded-degree from~\cite{testingFO}. 

\paragraph{Brief Description of the property $\classStruc{P}_{\zigzag}$.} $\classStruc{P}_{\zigzag}$ is the property of all bounded-degree $d$ $\sigma$-structures, which satisfy some first-order logic formula $\varphi_{\zigzag}$. On a high level, each structure $\struc{A}$ in the  property $\classStruc{P}_{\zigzag}$ is a hybridization of a sequence of expander graphs and a tree structure, where the expander graphs are constructed by the zig-zag product that was introduced in \cite{ZigZagProductIntroduction}.  
Slightly more precisely, each model of $\varphi_{\zigzag}$ is a rooted $k$-ary complete tree for some constant $k$, where the vertices on each level form an expander. In terms of logic language, for some constant $D>1$, we considered  
\begin{displaymath}
\sigma:=\{ \{E_{i,j}\}_{i,j \in \indexSetRotation},\{F_{k}\}_{k \in \indexSetH},R,\{L_k\}_{k\in \indexSetH}\},
\end{displaymath} 
where $E_{i,j}$, $F_{k}$, $R$ and $L_k$ are binary relation symbols for $i,j\in [D]^2$ and $k\in ([D]^2)^2$. We further use $F$ and  $E$ as an abbreviation to denote  $\bigcup_{i,j\in \indexSetRotation}E_{i,j}$ and $\bigcup_{k\in \indexSetH}F_{k}$. We   
defined an FO formula $\varphi_{\zigzag}$ such that  
\[
\varphi_{\zigzag}:=\varphi_{\operatorname{tree}}\land \varphi_{\operatorname{rotationMap}}\land \varphi_{\operatorname{base}}\land \varphi_{\operatorname{recursion}}, \text{ and } \classStruc{P}_{\zigzag}:=\{\mathcal{A}\in \classStruc{C}_{\sigma,d}\mid \mathcal{A} \models\varphi_{\zigzag} \},
\]
where  $\varphi_{\operatorname{tree}}, \varphi_{\operatorname{rotationMap}}, \varphi_{\operatorname{base}}, \varphi_{\operatorname{recursion}}$ are FO formulas which encode the tree structure (and degree regularity), rotation maps, base graph (with constant size) and recursive construction of expander graphs (via the zig-zag product). Note that for the construction we use some base graph $H$ which is given by its rotation map $\rot_H: \indexSetH \times [D]\rightarrow \indexSetH \times [D]$, which is  a special type of an encoding of a graph. 

 The precise formula is given in Appendix~\ref{app:A}. We will restate parts of the formula, whenever they are relevant in the proofs below. 

\subsection{Characterisation by neighbourhood profiles}\label{sec:charBy0Profiles}
Our aim in this section is to prove that  a minor variation of property $\classStruc{P}_{\zigzag}$ of relational structures can be written as a finite union of properties defined by $0$-profiles
of radius $2$. As the existence of a certain vertex cannot be expressed with a $0$-profile (see Observation~\ref{obs:expressingExOfType}) and $\varphi_{\zigzag}$ demands the existence of a certain vertex (the root vertex), the property $\classStruc{P}_{\zigzag}$ cannot be expressed  
in terms of $0$-profiles. However we 
 define a slight variation of the formula $\varphi_{\zigzag}$ which, as we will see later, can be expressed by $0$-profiles. Let 
\begin{displaymath}
	\varphi_{\zigzag}':=\varphi_{\operatorname{tree}}'\land \varphi_{\operatorname{rotationMap}}\land \varphi_{\operatorname{base}}\land \varphi_{\operatorname{recursion}},
\end{displaymath} 
where we obtain $\varphi_{\operatorname{tree}}'$ from $\varphi_{\operatorname{tree}}$ by replacing the subformula $\exists^{=1}x\varphi_{\operatorname{root}}(x)$ by $\exists^{\leq 1}x\varphi_{\operatorname{root}}(x)$, where $\varphi_{\operatorname{root}}(x):= \forall y \lnot F(y,x)$. We define the property 
\begin{displaymath}
\classStruc{P}_{\zigzag}':=\{\struc{A}\in \mathcal{C}_{\sigma,d}\mid \struc{A}\models \varphi_{\zigzag}'\}.
\end{displaymath} 
We denote the empty structure by $\struc{A}_{\emptyset}$ (\ie $\univ{A_{\emptyset}}=\emptyset$).

\begin{lemma}\label{lem:relaxedFormula}
	The properties $\classStruc{P}_{\zigzag}'$ and $\classStruc{P}_{\zigzag}\cup \{\struc{A}_{\emptyset}\}$ are equal.	 
\end{lemma}
To prove this we use the  following lemma \cite[Lemma 3.5]{testingFO}. 

\begin{lemma}[\cite{testingFO}]\label{lem:connectedComponentContainRoots}
	For  $\struc{A}\in \classStruc{C}_{\sigma,d}$ let $G_F^{\struc{A}}$ be the graph with vertex set $\univ{A}$ and edge set $\{\{a,b\}\mid (a,b)\in \rel{F}{A} \}$. If $\struc{A}\models \varphi_{\zigzag}$ then  $G_F^{\struc{A}}$ is connected.
\end{lemma}
\begin{proof}[Proof of Lemma~\ref{lem:relaxedFormula}]
	 We fist prove that $\classStruc{P}_{\zigzag}'\subseteq \classStruc{P}_{\zigzag}\cup \{\struc{A}_{\emptyset}\}$. Consider the formula $\tilde{\varphi}_{\zigzag}$ which is obtained from $\varphi_{\zigzag}$ by removing the subformula $\exists^{=1}x\varphi_{\operatorname{root}}(x)$. We use the following simple observation, which we will prove in Appendix~\ref{app:B}.
	 \begin{claim}\label{claim:closedUnderUnion}
		 Satisfying $\tilde{\varphi}_{\zigzag}$ is closed under disjoint unions on $\classStruc{C}_{\sigma,d}$. 
	 \end{claim}
	 Since $\struc{A}_{\emptyset}\in \classStruc{P}_{\zigzag}\cup \{\struc{A}_{\emptyset}\}$ it is sufficient to consider only non-empty structures in the following.  Therefore assume that there exists $\struc{A}\in \classStruc{C}_{\sigma,d}$ with $\univ{A}\not=\emptyset$ such that $\struc{A}\models \varphi_{\zigzag}'$ and $\struc{A}$ contains no element $a$ for which $\struc{A}\models \varphi_{\operatorname{root}}(a)$. Let $\struc{A}'\in \classStruc{C}_{\sigma,d}$ be any model of $\varphi_{\zigzag}$ with $\univ{A}\cap \univ{A'}=\emptyset$. Then 
	 $\struc{A}\cup \struc{A}'\models \tilde{\varphi}_{\zigzag}$ by Claim~\ref{claim:closedUnderUnion}. Furthermore, $\struc{A}\cup\struc{A}'\models \exists^{=1}x\varphi_{\operatorname{root}}(x)$, which implies $\struc{A}\cup\struc{A}'\models \varphi_{\zigzag}$. By construction $G_F^{\struc{A}\cup \struc{A}'}$ has more than one connected component as both $\univ{A}\not=\emptyset$ and $\univ{A'}\not=\emptyset$ and $\struc{A}\cup\struc{A}'$ is a disjoint union of $\struc{A}$ and $\struc{A}'$. Hence we obtain a contradiction to Lemma~\ref{lem:connectedComponentContainRoots}. 
	 Therefore every non-empty structure satisfying $\varphi_{\zigzag}'$ must satisfy $\exists^{=1}x\varphi_{\operatorname{root}}(x)$, and hence also $\varphi_{\zigzag}$. 
	 
	 Conversely, if $\struc{A}\in \classStruc{C}_{\sigma,d}$ is a model of $\varphi_{\zigzag}$ then $\struc{A}\models \exists^{=1}x\varphi_{\operatorname{root}}(x)$. This implies directly that $\struc{A}\models\exists^{\leq 1}x\varphi_{\operatorname{root}}(x)$ and hence $\struc{A}\models \varphi_{\zigzag}'$. Furthermore, $\struc{A}_{\emptyset}\in \classStruc{P}_{\zigzag}'$ as $\struc{A}\models \exists^{\leq 1}x\varphi_{\operatorname{root}}(x)$ and $\struc{A}\models \tilde{\varphi}_{\zigzag}$ as $\tilde{\varphi}_{\zigzag}$ is a conjunction of universally quantified formulas. Hence $ \classStruc{P}_{\zigzag}\cup \{\struc{A}_{\emptyset}\}\subseteq \classStruc{P}_{\zigzag}'$.
\end{proof}

We now define the $0$-profiles which express the property $\classStruc{P}_{\zigzag}'$.
For all $\sigma$-structures in $\classStruc{P}_{\zigzag}$ (all $\sigma$-structure in $\classStruc{P}_{\zigzag}'$ but $\struc{A}_{\emptyset}$) it is crucial that they are allowed to contain precisely one root element. Hence the neighbourhood profile describing $\classStruc{P}_{\zigzag}'$ must restrict the number of occurrences of the $2$-type of the root element. But since in $\classStruc{P}_{\zigzag}$, the root elements in different structures may have different $2$-types, 
we partition $\classStruc{P}_{\zigzag}$ into parts 
 $\classStruc{P}_1,\dots,\classStruc{P}_m$  by the $2$-type of the root element. 
 Note that the number $m$ of parts is constant as there are at most $n_{d,2,\sigma}$ $2$-types in total. For each of these parts we then define a neighbourhood profile  $\rho_k$ such that $\classStruc{P}_k\cup\{\struc{A}_{\emptyset}\}=\classStruc{P}_{\rho_k}$. 
 We would like to remark here that the roots of all but one structure in  $\classStruc{P}_{\zigzag}$ actually have the same $2$-types. 
However, proving this requires a detailed insight into the construction of $\mathcal{P}_{\zigzag}$, so we avoid this here and use the partition into finitely many parts instead. We now define the parts and corresponding profiles formally.

Assume without loss of generality that the $2$-types $\tau_{d,2,\sigma}^1,\dots, \tau_{d,2,\sigma}^{n_{d,2,\sigma}}$ of degree $d$ are ordered in such a way that for  $(\struc{B},b)\in \tau_{d,2,\sigma}^{k}$, it holds that $\struc{B}\models \varphi_{\operatorname{root}}(b)$  if and only if $k\in \{1,\dots,m\}$ for some $m\leq n_{d,2,\sigma}$.
For $k\in \{1,\dots,m\}$, let 
\[
\classStruc{P}_k:=\{\struc{A}\in \classStruc{P}_{\zigzag}\mid \text{ there is }a\in \univ{A}\text{ such that }(\mathcal{N}_2^{\struc{A}}(a),a)\in \tau_{d,2,\sigma}^k\}.
\] 
Since every $\struc{A}\in \classStruc{P}_{\zigzag}$ satisfies $ \exists^{=1}x\varphi_{\operatorname{root}}(x)$  we get that 
\[
\classStruc{P}_{\zigzag}'=\bigcup_{1\leq k\leq m}\classStruc{P}_k\cup \{\struc{A}_{\emptyset}\}
\]
 and this union is disjoint.   
Furthermore, for $k\in \{1,\dots,m\}$, let $I_k\subseteq \{1,\dots,n_{d,2,\sigma}\}$ be the set of indices $j$ such that there is a structure $\struc{A}\in \classStruc{P}_k$ and $a\in \univ{A}$ with $(\mathcal{N}_2^{\struc{A}}(a),a)\in\tau_{d,2,\sigma}^j$. 
For every $k\in \{1,\dots,m\}$ we define the $2$-neighbourhood profile   $\rho_k:\{1,\dots,n_{d,2,\sigma}\}\rightarrow \mathfrak{I}_{0}$ by
\begin{align*}
\rho_k(i):=
\begin{cases}
[0,1] &  \text{ if }i=k, \\
[0,\infty) & \text{ if }i\in I_k\setminus\{k\},\\
[0,0] & \text{ otherwise}.
\end{cases}
\end{align*}
To prove that these $0$-profiles of radius $2$ define the property $\classStruc{P}_{\zigzag}'$, the crucial observation is that for every element $a$ of some structure in $\classStruc{C}_{\sigma,d}$, the FO-formula $\varphi_{\zigzag}'$ only talks about elements of distance at most $2$ to $a$ (\ie $\varphi'_{\zigzag}$ is $2$-local). Hence 
the $2$-histogram vector of a structure already captures whether the structure satisfies $\varphi_{\zigzag}'$. We will now formally prove this.
\begin{lemma}\label{lem:neighbouhoodProfilOfPZigZag} 
It holds that $\classStruc{P}_{\zigzag}'=\bigcup_{1\leq k \leq m}\classStruc{P}_{\rho_k}$. 
\end{lemma}
\begin{proof}
We first prove that $\classStruc{P}_{\zigzag}'\subseteq \bigcup_{1\leq k \leq m}\classStruc{P}_{\rho_k}$. First note that trivially $\struc{A}_{\emptyset}\in \bigcup_{1\leq k \leq m}\classStruc{P}_{\rho_k}$. Now assume $\struc{A}\in \classStruc{P}_{\zigzag}$. This implies that there is $k\in \{1,\dots,m\}$ such that $\struc{A}\in \classStruc{P}_k$.  By construction we have that for every $a\in \struc{A}$, there is  $i\in I_k$  such that $(\mathcal{N}_2^{\struc{A}}(a),a)\in \tau_{d,2,\sigma}^i$. Furthermore, since $\struc{A}\models \varphi_{\zigzag}$, we have that $\struc{A}\models \exists^{=1}x\varphi_{\operatorname{root}}(x)$, and that there can be at most one $a\in \univ{A}$   such that $(\mathcal{N}_2^{\struc{A}}(a),a)\in \tau_{d,2,\sigma}^k$.   Therefore $\struc{A}\in \classStruc{P}_{\rho_k}$. \\ 
		
To prove $\bigcup_{1\leq k \leq m}\classStruc{P}_{\rho_k}\subseteq \classStruc{P}_{\zigzag}'$, we  prove that  
	every structure in $\bigcup_{1\leq k \leq m}\classStruc{P}_{\rho_k}$  
	must satisfy $\varphi_{\zigzag}'$. We will prove that every $\struc{A}\in \bigcup_{1\leq k \leq m}\classStruc{P}_{\rho_k}$ satisfies $\varphi_{\operatorname{recursion}}$, and refer for the proof that $\struc{A}$ satisfies $\varphi_{\operatorname{tree}}'\land \varphi_{\operatorname{rotationMap}}\land \varphi_{\operatorname{base}}$ to  Claim~\ref{claim:satisfyingTree}, Claim~\ref{claim:satisfyingRotationMap} and Claim~\ref{claim:satisfyingBase} in Appendix~\ref{app:B}. Note that $\struc{A}_{\emptyset}\models \varphi_{\zigzag}'$ by Lemma~\ref{lem:relaxedFormula} and hence we exclude $\struc{A}_{\emptyset}$ in the following.
	
	\begin{claim}\label{claim:satisfyingRecursion}
		Every structure $\struc{A}\in \bigcup_{1\leq k \leq m}\classStruc{P}_{\rho_k}\setminus \{\struc{A}_{\emptyset}\}$ satisfies $\varphi_{\operatorname{recursion}}$.
	\end{claim}
	\begin{proof}
		Let $\struc{A}\in \bigcup_{1\leq k \leq m}\classStruc{P}_{\rho_k}\setminus \{\struc{A}_{\emptyset}\}$. Then there is a $k\in \{1,\dots,m\}$ such that $\struc{A}\in \classStruc{P}_{\rho_k}$. 
		
		By definition, $\varphi_{\operatorname{recursion}}:= \forall x \forall z\big(\varphi(x,z)\lor \psi(x,z)\big)$ (see Appendix~\ref{app:A}),  where 
		\begin{align*}
		\varphi(x,z):=&\lnot \exists y F(x,y)\land \lnot \exists y F(z,y) \text{ and }\\
		\psi(x,z):=&\bigwedge_{\substack{k_1',k_2'\in \indexSetRotation\\\ell_1',\ell_2'\in \indexSetRotation}}\bigg(\exists y \big[E_{k_1',\ell_1'}(x,y)\land E_{k_2',\ell_2'}(y,z)\big]\rightarrow \\&
		\bigwedge_{\substack{i,j,i',j'\in [D], k,\ell\in \indexSetH\\\rot_H(k,i)=((k_1', k_2'),i')\\ \rot_H((\ell_2', \ell_1'),j)=(\ell,j')}}\exists x'\exists z'\big[ F_k(x,x')\land F_\ell(z,z')\land
		E_{(i,j),(j',i')}(x',z')\big]\bigg).
		\end{align*}
		Let $a,c\in \univ{A}$. Assume first that there is $b\in \univ{A}$ with $(a,b)\in \rel{F}{\struc{A}}$.  Hence $\struc{A}\not\models \varphi(a,c)$. Since $\varphi_{\operatorname{recursion}}:= \forall x \forall z\big(\varphi(x,z)\lor \psi(x,z)\big)$ we aim to prove $\struc{A}\models \psi(a,c)$.  
		By construction of $\rho_k$, there is an $i\in I_{k}$ such that $(\mathcal{N}_2^{\struc{A}}(a),a)\in \tau_{d,2,\sigma}^i$. Therefore there is a structure $\other{\struc{A}}\models \varphi_{\zigzag}$ and $\other{a}\in \univ{\other{A}}$ such that $(\mathcal{N}_2^{\struc{A}}(a),a)\cong (\mathcal{N}_2^{\other{\struc{A}}}(\other{a}),\other{a})$.  Let $f$ be an isomorphism from $(\mathcal{N}_2^{\struc{A}}(a),a)$ to $(\mathcal{N}_2^{\other{\struc{A}}}(\other{a}),\other{a})$. Since $b\in N_2^{\struc{A}}(a)$, we get that $f(b)$ is defined. Since $f$ is an isomorphism mapping $a$ onto $\other{a}$, we have that $(a,b)\in \rel{F}{\struc{A}}$ implies that $(\other{a},f(b))\in \rel{F}{\other{\struc{A}}}$. Hence $\other{\struc{A}}\not\models \varphi(\other{a},\other{c})$, for every $\other{c}\in \univ{\other{A}}$. But since $\other{\struc{A}}\models \varphi_{\operatorname{recursion}}$, as $\other{\struc{A}}\models \varphi_{\zigzag}$, this shows that $\other{\struc{A}}\models \psi(\other{a},\other{c})$ for every $\other{c}\in \univ{\other{A}}$. 
		
		Let $k_1',k_2'\in \indexSetRotation$ and $\ell_1',\ell_2'\in \indexSetRotation$ be indices such that there is  $b'\in \univ{A}$ with $(a,b')\in \rel{E_{k_1',\ell_1'}}{\struc{A}}$ and $(b',c)\in \rel{E_{k_2',\ell_2'}}{\struc{A}}$. 
		Since $b',c\in N_2^{\struc{A}}(a)$, by assumption we get that $f(b')$ and $f(c)$ are defined. Furthermore, $(a,b')\in \rel{E_{k_1',\ell_1'}}{\struc{A}}$ and $(b',c)\in \rel{E_{k_2',\ell_2'}}{\struc{A}}$ imply that $(\other{a},f(b'))\in \rel{E_{k_1',\ell_1'}}{\other{\struc{A}}}$ and $(f(b'),f(c))\in \rel{E_{k_2',\ell_2'}}{\other{\struc{A}}}$, since $f$ is an isomorphism mapping $a$ onto $\other{a}$. We proved in the previous paragraph that $\other{\struc{A}}\models \psi(\other{a},f(c))$. Hence we can conclude that for all indices $i,j,i',j'\in [D]$,  $k,\ell\in \indexSetH$ for which $\rot_H(k,i)=((k_1', k_2'),i')$ and  $\rot_H((\ell_2', \ell_1'),j)=(\ell,j')$, there are elements $\other{a}',\other{c}'\in \univ{\other{A}}$ such that $ (\other{a},\other{a}')\in \rel{F_k}{\other{\struc{A}}}$, $(f(c),\other{c}')\in \rel{F_\ell}{\other{\struc{A}}}$, and $(\other{a}',\other{c}')\in \rel{E_{(i,j),(j',i')}}{\other{\struc{A}}}$. Since  $\other{a}',\other{c}'\in N_2^{\other{\struc{A}}}(\other{a})$, we get that $a':=f^{-1}(\other{a}')$ and $c':=f^{-1}(\other{c}')$ are defined. Furthermore, we get that $ (a,a')\in \rel{F_k}{\struc{A}}$, $(c,c')\in \rel{F_\ell}{\struc{A}}$ and $(a',c')\in \rel{E_{(i,j),(j',i')}}{\struc{A}}$. This proves that $\struc{A}\models \psi(a,c)$.\\
		
		In the case that there is  $b\in \univ{A}$ with $(c,b)\in \rel{F}{\struc{A}}$, we can prove similarly that $\struc{A}\models \psi(a,c)$, by considering that there exist $\other{\struc{A}}\models \varphi_{\zigzag}$ and $\other{c}\in \univ{\other{A}}$ such that $(\mathcal{N}_2^{\struc{A}}(a),c)\cong (\mathcal{N}_2^{\other{\struc{A}}}(\other{c}),\other{c})$ by construction of $\rho_k$. Finally if there is no $b\in \univ{A}$ such that $(a,b)\in \rel{F}{\struc{A}}$ or $(c,b)\in \rel{F}{\struc{A}}$ then $\struc{A}\models \varphi(a,c)$. Since this covers every case we get that $\struc{A}\models \varphi_{\operatorname{recursion}}$.
	\end{proof}
	Assume $\struc{A}\in \bigcup_{1\leq k \leq m}\classStruc{P}_{\rho_k}$. As proved in Claims~\ref{claim:satisfyingTree}, \ref{claim:satisfyingRotationMap}, \ref{claim:satisfyingBase} and \ref{claim:satisfyingRecursion} this implies that $\struc{A}\models \varphi_{\operatorname{tree}}'$, $\struc{A}\models \varphi_{\operatorname{rotationMap}}$, $\struc{A}\models \varphi_{\operatorname{base}}$ and $\struc{A}\models \varphi_{\operatorname{recursion}}$. Since $\varphi_{\zigzag}'$ is a conjunction of these formulas, we get $\struc{A}\models \varphi_{\zigzag}'$ and hence $\struc{A}\in \classStruc{P}_{\zigzag}'$.
\end{proof}
\subsection{A local reduction from relational structures to graphs}\label{sec:localreduction}
In this section we will define our graph property $\graphProp$ by giving a reduction from the property $\classStruc{P}_{\zigzag}'$ and argue that $\graphProp$ is GSF-local while not testable. To do so, we show that this reduction is `local' which preserves the testability of these two properties. 
\paragraph{Local reduction} We first introduce the following notion of local reduction between two property testing models. In the following, when the context is clear, we will use $\mathcal{C}$ to denote both a class of structure and the corresponding property testing model, which can be either the bounded-degree model for graphs or bounded-degree model for relational structures.   
\begin{definition}[Local reduction]
	Let $\mathcal{C},\mathcal{C'}$ be two property testing models 
	and let $\mathcal{P}\subseteq\mathcal{C}$, $\mathcal{P}'\subseteq\mathcal{C'}$ be two properties. We say that a function $f:\mathcal{C}\rightarrow \mathcal{C'}$ is a local reduction from $\mathcal{P}$ to $\mathcal{P}'$ if there are constants $c_1,c_2\in \mathbb{N}_{\geq 1}$ such that for every $X\in \mathcal{C}$ the following properties hold.
	\begin{enumerate}
		\item If $X\in \mathcal{P}$ then $f(X)\in \mathcal{P}'$.
		\item If $X$ is $\epsilon$-far from $\mathcal{P}$ then $f(X)$ is $(\epsilon/c_1)$-far from $\mathcal{P}'$.
		\item For every query to $f(X)$ we can adaptively\footnote{By adaptively computing queries we mean that the selection of the next query may depend on the answer to the previous query.} compute $c_2$ queries such that the answer to the query to $f(X)$ can be computed  from the answers to the $c_2$ queries to $X$.
	\end{enumerate}
\end{definition}
The following lemma is known. 
\begin{lemma}[Theorem 7.14 in \cite{goldreich2017introduction}]\label{lem:localReduction}
	Let $\mathcal{C},\mathcal{C'}$ be two property testing models, $\mathcal{P}\subseteq\mathcal{C}$, $\mathcal{P}'\subseteq\mathcal{C'}$ be two properties	and $f$ a local reduction from $\mathcal{P}$ to $\mathcal{P}'$. If $\mathcal{P}'$ is testable then so is $\mathcal{P}$.
\end{lemma}
\paragraph{Construction of the graph property} Now we construct a property $\graphProp$ from the property $\classStruc{P}_{\zigzag}'$. We obtain this graph property as $f(\classStruc{P}_{\zigzag}')$ by defining a map $f:\classStruc{C}_{\sigma,d}\rightarrow \mathcal{C}_d$. To define $f$ we  introduce a distinct arrow-graph gadget for every relation in $\sigma$ (\ie for every edge colour). The map $f$ then replaces every tuple in a certain relation (every coloured edge) by the respective arrow-graph gadget. We further prove that this replacement operation defines a local reduction $f$ from $\classStruc{P}_{\zigzag}'$ to $\graphProp$. Recall that a local reduction is a function maintaining distance that can be simulated locally by queries. Since by Lemma~\ref{lem:localReduction} local reductions preserve testability, we use the local reduction from $\classStruc{P}_{\zigzag}'$ to $\graphProp$ to obtain non-testability of the property  $\graphProp$ from the non-testability of $\classStruc{P}_{\zigzag}'$. We will now define $f$ formally.

Let $\ell$ be the number of relations (the number of edge colours) in $\sigma$. 
We first introduce the different types of arrow-graph gadgets we need to define the local reduction.
	For $1\leq k\leq \ell$, we let $H_k$ 
	 be the graph with vertex set $V(H_k):=\{a_1,\dots,a_{2\ell+2},b_1,b_2\}$ and edge set $E(H_k):=\{\{a_i,a_{i+1}\}\mid 1\leq i\leq 2\ell+1 \}\cup\{\{a_{\ell+1+k},b_j\}\mid j\in\{1,2\}\}$.	We call $H_k$ a \emph{$k$-arrow}. For any graph $G$ and vertices $v,w\in V(G)$, we say that there is a $k$-arrow from $v$ to $w$, denoted $v\xrightarrow{k}w$, if there are $2\ell+2$ vertices $v_2,\dots,v_{2\ell+1},w_1,w_2\in V(G)$ and an isomorphism  $g:H_k\rightarrow \mathcal{N}_1^G(v_2,\dots,v_{2\ell+1},w_1,w_2)$ such that $g(a_1)=v$ and $g(a_{2\ell+2})=w$. We now define a second arrow gadget. For $1\leq k\leq \ell$, we let $L_k$ be the graph with vertex set $V(L_k):=\{a_1,\dots,a_{\ell+1},b\}$ and edge set $E(L_k):=\{\{a_i,a_{i+1}\}\mid 1\leq i\leq \ell \}\cup\{\{a_{k},b\}\}$.	We call $L_k$ a \emph{$k$-loop}. For any graph $G$ and vertex $v\in V(G)$, we say that there is a $k$-loop at $v$, denoted $v\xrightarrow{k}v$, if there are $\ell+1$ vertices $v_1,\dots,v_{\ell},w\in V(G)$ and an isomorphism  $g:L_k\rightarrow \mathcal{N}_1^G(v_1,\dots,v_{\ell},w)$ such that $g(a_{\ell+1})=v$. Finally we let $H_{\bot}$ be the graph with vertex set $V(H_{\bot}):=\{a_1,\dots,a_{\ell+1},b\}$ and edge set $E(H_{\bot}):=\{\{a_i,a_{i+1}\}\mid 1\leq i\leq \ell \}\cup\{\{a_{i},b\}\mid i\in\{1,2\}\}$.	We call $H_{\bot}$ a \emph{non-arrow}. For any graph $G$ and vertex $v\in V(G)$, we say that there is a non-arrow at $v$, denoted $v\not\rightarrow$, if there are $\ell+1$ vertices $v_1,\dots,v_{\ell},w\in V(G)$ and an isomorphism  $g:H_{\bot}\rightarrow N_1^G(v_1,\dots,v_{\ell},w)$ such that $g(a_{\ell+1})=v$.

We now define a function $f:\classStruc{C}_{\sigma,d}\rightarrow \mathcal{C}_d$ by	$f(\struc{A}):=G_\struc{A}$, where $G_{\struc{A}}$ is the graph on vertex set $V(G_{\struc{A}}):=\univ{A}\cup \{v^k_{a,i},w_{a,i}\mid  1\leq i\leq d, a\in \univ{A},1\leq k\leq \ell\}$ and edge set 
\begin{align*}
E(G_{\struc{A}}):=&\Big\{\{a,v^\ell_{a,i}\}\mid a\in \univ{A}, 1\leq i\leq d\Big\}\cup \Big\{\{v^k_{a,i},v^{k+1}_{a,i}\}\mid 1\leq k\leq \ell-1,a\in \univ{A}, 1\leq i\leq d\Big\}\\
\cup &\Big\{\{v^{k}_{b,j},w_{b,j}\},\{v^{k}_{b,j},w_{a,i}\},\{v^\ell_{a,i},v^\ell_{b,j}\}\mid a\not=b, \operatorname{ans}(a,i)=\operatorname{ans}(b,j)=(k,a,b) \Big\}\\
\cup & \Big\{ \{v^{k}_{a,i},w_{a,i}\}\mid \operatorname{ans}(a,i)=(k,a,a)  \Big\}\cup \Big\{ \{v^{1}_{a,i},w_{a,i}\},\{v^{2}_{a,i},w_{a,i}\}\mid \operatorname{ans}(a,i)=\bot  \Big\},
\end{align*}
where $\operatorname{ans}(a,i)=(k,a,b)$ denotes that the $i$-th tuple of $a$ is $(a,b)$ and is in the $k$-th relation. 
Hence $G_{\struc{A}}$ is defined in such a way that if $(a,b)$ is a tuple in the $k$-th relation of $\sigma$ in $\struc{A}$, then $a\xrightarrow{k}b$ in $G_{\struc{A}}$,  and $a$ has a non-arrow for every $i$ satisfying that $\operatorname{ans}(a,i)=\bot$ for every $k$. For illustration see Figure~\ref{fig:path}.

Now we  define property  $\graphProp:=\{f(\struc{A})\mid \struc{A}\in \classStruc{P}_{\zigzag}'\}\subseteq \mathcal{C}_d$.
\begin{figure*}
	\centering
	\begin{minipage}{0.47\linewidth}
		\centering
		\begin{tikzpicture}
		\tikzstyle{ns1}=[line width=1.2]
		\tikzstyle{ns2}=[line width=1.2]
		\definecolor{C1}{RGB}{1,1,1}
		\definecolor{C2}{RGB}{0,0,170}
		\definecolor{C3}{RGB}{251,86,4}
		\definecolor{C4}{RGB}{50,180,110}
		\def \dist {0.7}
		\def \heightWriting {0.4}
		\node[circle,fill=white,inner sep=0pt, minimum width=0pt] (15) at (0,-0.3) {};
		\node[draw,circle,fill=black,inner sep=0pt, minimum width=8pt] (0) at (0,0) {};
		\node[draw,circle,fill=black,inner sep=0pt, minimum width=5pt] (1) at (\dist,0) {};	
		\node[draw,circle,fill=black,inner sep=0pt, minimum width=5pt] (3) at (5-1.5*\dist,0) {};
		\node[draw,circle,fill=black,inner sep=0pt, minimum width=5pt] (4) at (5-0.5*\dist,0) {};
		\node[draw,circle,fill=black,inner sep=0pt, minimum width=5pt] (5) at (5+0.5*\dist,-1) {};	
		\node[draw,circle,fill=black,inner sep=0pt, minimum width=5pt] (13) at (5+0.5*\dist,1) {};
		\draw[ns1] (0)--(1);
		\draw[ns1] (13)--(5);
		\draw[ns1] (4)--(13);
		\draw[ns1] (3)--(4)--(5);
		\draw[ns1,dotted] (1)--(3);

		\node[minimum height=10pt,inner sep=0] at (0,\heightWriting+0.05) {$a$};
		\node[minimum height=10pt,inner sep=0] at (\dist,\heightWriting) {$v_{a,i}^\ell$};
		\node[minimum height=10pt,inner sep=0] at (5-0.5*\dist-0.2,\heightWriting) {$v_{a,i}^2$};
		\node[minimum height=10pt,inner sep=0] at (5+0.5*\dist+0.4,-1) {$v_{a,i}^1$};
		\node[minimum height=10pt,inner sep=0] at (5+0.5*\dist+0.4,1) {$w_{a,i}$};				
		
		\end{tikzpicture}
		\subcaption{Case $\operatorname{ans}(a,i)=\bot$} 	
	\end{minipage}
	\hfill
	\begin{minipage}{0.47\linewidth}
		\centering
		\begin{tikzpicture}
		\tikzstyle{ns1}=[line width=1.2]
		\tikzstyle{ns2}=[line width=1.2]
		\definecolor{C1}{RGB}{1,1,1}
		\definecolor{C2}{RGB}{0,0,170}
		\definecolor{C3}{RGB}{251,86,4}
		\definecolor{C4}{RGB}{50,180,110}
		\def \dist {0.7}
		\def \heightWriting {0.4}
		\node[circle,fill=white,inner sep=0pt, minimum width=0pt] (15) at (0,-0.3) {};
		\node[draw,circle,fill=black,inner sep=0pt, minimum width=8pt] (0) at (0,0) {};
		\node[draw,circle,fill=black,inner sep=0pt, minimum width=5pt] (1) at (\dist,0) {};	
		\node[draw,circle,fill=black,inner sep=0pt, minimum width=5pt] (2) at (2.7,0) {};
		\node[draw,circle,fill=black,inner sep=0pt, minimum width=5pt] (3) at (2.7+\dist,0) {};
		\node[draw,circle,fill=black,inner sep=0pt, minimum width=5pt] (4) at (2.7+2*\dist,0) {};
		\node[draw,circle,fill=black,inner sep=0pt, minimum width=5pt] (5) at (5+0.5*\dist,0) {};	
		\node[draw,circle,fill=black,inner sep=0pt, minimum width=5pt] (13) at (2.7+2*\dist,1.3) {};
		\draw[ns1] (0)--(1);
		\draw[ns1] (2)--(3);
		\draw[ns1] (3)--(4);
		\draw[ns1] (3)--(13);
		\draw[ns1,dotted] (1)--(2);	
		\draw[ns1,dotted] (4)--(5);			
		
		\node[minimum height=10pt,inner sep=0] at (0,\heightWriting+0.05) {$a$};
		\node[minimum height=10pt,inner sep=0] at (\dist,\heightWriting) {$v_{a,i}^\ell$};
		\node[minimum height=10pt,inner sep=0] at (2.7+\dist-0.17,\heightWriting) {$v_{a,i}^k$};
		\node[minimum height=10pt,inner sep=0] at (5+0.5*\dist-0.1,\heightWriting) {$v_{a,i}^1$};
		\node[minimum height=10pt,inner sep=0] at (2.7+2*\dist,1.3+\heightWriting) {$w_{a,i}$};				
		
		\end{tikzpicture} 
		\subcaption{Case $\operatorname{ans}(a,i)=(k,a,a)$}
	\end{minipage}
	\begin{minipage}{1\linewidth}
		\centering
		\begin{tikzpicture}
		\tikzstyle{ns1}=[line width=1.2]
		\tikzstyle{ns2}=[line width=1.2]
		\definecolor{C1}{RGB}{1,1,1}
		\definecolor{C2}{RGB}{0,0,170}
		\definecolor{C3}{RGB}{251,86,4}
		\definecolor{C4}{RGB}{50,180,110}
		\def \dist {0.7}
		\def \heightWriting {0.4}
		\node[circle,fill=white,inner sep=0pt, minimum width=0pt] (15) at (0,3) {};
		\node[draw,circle,fill=black,inner sep=0pt, minimum width=8pt] (0) at (0,0) {};
		\node[draw,circle,fill=black,inner sep=0pt, minimum width=5pt] (1) at (\dist,0) {};	
		\node[draw,circle,fill=black,inner sep=0pt, minimum width=5pt] (4) at (5-0.5*\dist,0) {};
		\node[draw,circle,fill=black,inner sep=0pt, minimum width=5pt] (5) at (5+0.5*\dist,0) {};
		\node[draw,circle,fill=black,inner sep=0pt, minimum width=5pt] (7) at (7,0) {};
		\node[draw,circle,fill=black,inner sep=0pt, minimum width=5pt] (8) at (7+\dist,0) {};
		\node[draw,circle,fill=black,inner sep=0pt, minimum width=5pt] (9) at (7+2*\dist,0) {};		
		\node[draw,circle,fill=black,inner sep=0pt, minimum width=5pt] (11) at (10-\dist,0) {};
		\node[draw,circle,fill=black,inner sep=0pt, minimum width=8pt] (12) at (10,0) {};	
		\node[draw,circle,fill=black,inner sep=0pt, minimum width=5pt] (13) at (7,1.3) {};
		\node[draw,circle,fill=black,inner sep=0pt, minimum width=5pt] (14) at (7,-1.3) {};
		\draw[ns1] (0)--(1);
		
		\draw[ns1] (4)--(5);
		\draw[ns1] (7)--(8);
		\draw[ns1] (8)--(13);
		\draw[ns1] (8)--(14);
		\draw[ns1] (8)--(9);
		\draw[ns1] (11)--(12);
		\draw[ns1,dotted] (1)--(4);	
		\draw[ns1,dotted] (5)--(7);
		\draw[ns1,dotted] (9)--(11);

		\node[minimum height=10pt,inner sep=0] at (0,\heightWriting+0.05) {$a$};
		\node[minimum height=10pt,inner sep=0] at (\dist,\heightWriting) {$v_{a,i}^\ell$};
		\node[minimum height=10pt,inner sep=0] at (5-0.5*\dist-0.1,\heightWriting) {$v_{a,i}^1$};
		\node[minimum height=10pt,inner sep=0] at (7,1.3+\heightWriting) {$w_{a,i}$};		
		\node[minimum height=10pt,inner sep=0] at (10,\heightWriting+0.07) {$b$};
		\node[minimum height=10pt,inner sep=0] at (5+0.5*\dist+0.1,\heightWriting) {$v_{b,j}^1$};
		\node[minimum height=10pt,inner sep=0] at (7+\dist+0.2,\heightWriting) {$v_{b,j}^k$};
		\node[minimum height=10pt,inner sep=0] at (10-\dist,\heightWriting) {$v_{b,j}^\ell$};
		\node[minimum height=10pt,inner sep=0] at (7,-1.3-\heightWriting) {$w_{b,j}$};			
		
		\end{tikzpicture} 
		\subcaption{Case $\operatorname{ans}(a,i)=\operatorname{ans}(b,j)=(k,a,b)$}
	\end{minipage}
	\caption{Different types of arrows in $G_{\struc{A}}$.}\label{fig:path}
	
\end{figure*}
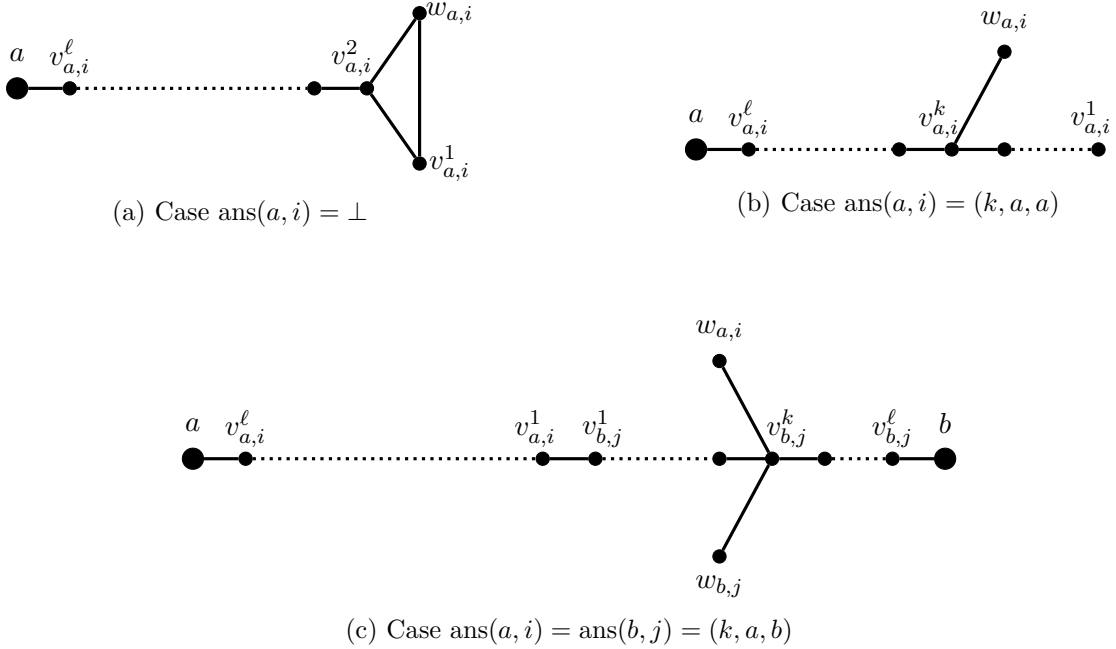
\begin{lemma}\label{lem:local_reduction}
The map	$f$ is a local reduction from $\classStruc{P}_{\zigzag}'$ to $\graphProp$.
\end{lemma}
\begin{proof}
	First note that for any $\struc{A}\in \classStruc{P}_{\zigzag}'$, we have that $f(\struc{A})\in \graphProp$ by definition. 
	
	Now let $c_1=2d+2d^2\ell$. We prove that  if $\struc{A}\in \classStruc{C}_{\sigma,d}$ is $\epsilon$-far from $\classStruc{P}_{\zigzag}'$ then $f(\struc{A})$ is $\epsilon/c_1$-far from $\graphProp$ by contraposition. Therefore assume that $f(\struc{A})=:G_{\struc{A}}$ is not $\epsilon/c_1$-far from $\graphProp$ for some $\struc{A}\in \classStruc{C}_{\sigma,d}$. Then there is a set $E\subseteq \{e\subseteq V(G_{\struc{A}})\mid |e|=2 \}$ of size at most $\epsilon d |V(G_{\struc{A}})|/c_1$,  and a graph $G\in \graphProp$ such that  $G$ is obtained from $G_{\struc{A}}$ by modifying the tuples in $E$. By definition of $\graphProp$, there is a structure $\struc{A}_G\in \classStruc{P}_{\zigzag}'$ such that $f(\struc{A}_G)=G$. First note that $|\univ{A_G}|=|\univ{A}|$, as $(1+d\ell)|\univ{A}|=|V(G_{\struc{A}})|=|V(G)|=(1+d\ell)|\univ{A_G}|$. Hence there must be a set $R$ of tuples that need to be modified to make $\struc{A}$ isomorphic to $\struc{A}_G$. First note that $R$ cannot contain a tuple $(a,b)$ where $\{a,v^k_{a,i},w_{a,i},b,v_{b,i}^k,w_{b,i}\mid  1\leq i\leq d, 1\leq k\leq \ell\}\cap e=\emptyset$ for every $e\in E$. This is because if $(a,b)$ is a tuple in $\struc{A}$, then $a\xrightarrow{k}b$ for some $k$ in $G_{\struc{A}}$. But since $\{a,v^k_{a,i},w_{a,i},b,v_{b,i}^k,w_{b,i}\mid  1\leq i\leq d, 1\leq k\leq \ell\}\cap e=\emptyset$ for every $e\in E$, we have that $a\xrightarrow{k}b$ in $G$. But then $(a,b)$ must be a tuple in $\struc{A}_G$, and hence $(a,b)$ cannot be in $R$. The same argument works when assuming that $(a,b)$ is a tuple in $\struc{A}_G$. Since for every $e\in E$, there are at most $2d$ tuples $(a,b)$ such that  $\{a,v^k_{a,i},w_{a,i},b,v_{b,i}^k,w_{b,i}\mid  1\leq i\leq d, 1\leq k\leq \ell\}\cap e\not=\emptyset$, we get that 
	\begin{displaymath}
	|R|\leq 2d \epsilon d |V(G_{\struc{A}})|/c_1=2d(1+d\ell)\epsilon d|\univ{A}|/c_1=\epsilon d |\univ{A}|.
	\end{displaymath}
	Hence $\struc{A}$ is not $\epsilon$-far to being in $\classStruc{P}_{\zigzag}'$.

	Let $c_2:= d+1$. Let $\struc{A}\in \classStruc{C}_{\sigma,d}$ and $G_{\struc{A}}:=f(\struc{A})$. Note that any $a\in \univ{A}$  is adjacent in $G_{\struc{A}}$ to $v_{a,i}^\ell$, for every $1\leq i\leq d$. Hence any neighbour query in $G_{\struc{A}}$ to $a$ 
	can be answered without querying $\struc{A}$.  Assume $v\in \{v_{a,i}^k,w_{a,i}\mid 1\leq k\leq \ell \}$ for some $a\in \univ{A}$ and some $1\leq i\leq d$. Then we can determine all neighbours of $v$ by querying $(a,i)$ and further if  $\operatorname{ans}(a,i)\not=\bot$ and $\operatorname{ans}(a,i)=(k,a,b)$, then we need to query $(b,j)$ for every $1\leq j\leq d$. Hence we can determine the answer to any query to $G_{\struc{A}}$ by making $c_2$ queries to $\struc{A}$. This proves that $f$ is a local reduction  from $\classStruc{P}_{\zigzag}'$ to $\graphProp$.
\end{proof}
We remark that $\graphProp$ is a simpler version of the  simple graph property defined in \cite{testingFO} where extra care had to be taken to define degree-regular graphs.

\subsection{The graph property is GSF-local}  
Let $\graphProp$ be the graph property as defined in Section~\ref{sec:localreduction}. We now show that $\graphProp$ is GSF-local. 
\begin{lemma}\label{lemma:graphproperty_gsf_local}
The graph property $\graphProp$ is GSF-local.
\end{lemma}
\begin{proof}
For this we will prove that $\graphProp$ is equal to a finite union of properties defined by $0$-profiles, and then use Theorem~\ref{thm:subsetOfFOIsLocal} to prove that $\graphProp$ is GSF-local. We define the $0$-profiles for $\graphProp$ in a very similar way to the relational structure case, and then use the description of $\classStruc{P}_{\zigzag}'$ by $0$-profiles shown in Lemma~\ref{lem:neighbouhoodProfilOfPZigZag}. To this end, assume that the $4\ell+2$-types $\tau_{d,4\ell+2}^1,\dots,\tau_{d,4\ell+2}^{n_{d,4\ell+2}}$ are ordered in such a way that $(\mathcal{N}_{4\ell+2}^{f(\struc{B})}(b),b)\in \tau_{d,4\ell+2}^{k}$, for every $k\in \{1,\dots,m\}$ and $(\struc{B},b)\in \tau_{d,2,\sigma}^{k}$, where $m$ is the number of parts of the partition of $\classStruc{P}_{\zigzag}$ defined in Subsection~\ref{sec:charBy0Profiles}. For $k\in \{1,\dots,m\}$, let  $\hat{I}_{k}$ be the set of indices $i$ such that there is $\struc{A}\in \classStruc{P}_k$, and $v\in V(f(\struc{A}))$ for which $(\mathcal{N}_{4\ell+2}^{f(\struc{A})}(v),v)\in \tau_{d,4\ell+2}^i$. Let $\hat{\rho_k}:\{1,\dots,n_{d,4\ell+2}\}\rightarrow \mathfrak{I}_{0}$ be defined by
\begin{align*}
\hat{\rho_k}(i):=
\begin{cases}
[0,1] &  \text{ if }i=k, \\
[0,\infty) & \text{ if }i\in \hat{I}_{k}\setminus\{k\},\\
[0,0] & \text{ otherwise}.
\end{cases}
\end{align*}
\begin{claim}\label{claim:profilOfGraphProperty}
	It holds that $\graphProp=\bigcup_{1\leq k \leq m}\classStruc{P}_{\hat{\rho}_k}$.
\end{claim}
\begin{proof}
	First we prove $\graphProp\subseteq \bigcup_{1\leq k \leq m}\classStruc{P}_{\hat{\rho}_k}$. Assume $G\in \graphProp$ and let $\struc{A}\in \classStruc{P}_{\zigzag}'$ be a structure such that $G=f(\struc{A})$. If $\struc{A}=\struc{A}_{\emptyset}$ then clearly $G\in \bigcup_{1\leq k \leq m}\classStruc{P}_{\hat{\rho}_k}$. Hence assume $\struc{A}\not=\struc{A}_{\emptyset}$. Then $\struc{A}\in \classStruc{P}_k$ for some $k\in \{1,\dots,m\}$. By the construction of $\hat{I}_{k}$ we know that for every $v\in V(G)$ we have  $(\mathcal{N}_{4\ell+2}^G(v),v)\in \tau_{d,4\ell+2}^i$ for some $i\in \hat{I}_{k}$. Furthermore, since $\struc{A}\in \classStruc{P}_k$ there is at most  one $a\in \univ{A}$ with $(\mathcal{N}_{2}^{\struc{A}}(a),a)\in \tau_{d,2,\sigma}^{k}$. This implies directly that there can be at most one vertex $v\in V(G)$ with  $(\mathcal{N}_{4\ell+2}^G(v),v)\in \tau_{d,4\ell+2}^{k}$ and hence $G\in \classStruc{P}_{\hat{\rho}}$. \\
	
	Now we prove that $\bigcup_{1\leq k \leq m}\classStruc{P}_{\hat{\rho}_k}\subseteq \graphProp$. Let $G\in \bigcup_{1\leq k \leq m}\classStruc{P}_{\hat{\rho}_k}$ and let $k\in \{1,\dots,m\}$ be an index such that $G\in \classStruc{P}_{\hat{\rho}_k}$. 
	
	First note that every model of $\varphi_{\zigzag}$ is $d$ regular for some large $d$. Then for any $\struc{A}\models \varphi_{\zigzag}$, every vertex in $f(\struc{A})$ has either degree $\leq 4$ or degree $d$ . Since every structure in $\classStruc{P}_{\zigzag}'$ apart from the empty structure $\struc{A}_{\emptyset}$ is a model of $\varphi_{\zigzag}$, this implies that every vertex in any graph $G'\in \graphProp$ has degree $\leq 4$ or degree $d$. Since for every $i$ for which $\hat{\rho}(i)\not=[0,0]$, there is a graph $G'\in \graphProp$ and $v\in V(G')$ such that $(\mathcal{N}_{4\ell+2}^{G'}(v),v)\in \tau_{d,4\ell+2}^i$, we get that every vertex in $G$ has to have degree $\leq 4 $ or degree $d$. Using this argument further, we get that every vertex $v\in V(G)$ of degree $\leq 4$ has to be contained in the $(\ell+1)$-neighbourhood of a vertex of degree $d$, and that the $(2\ell+1)$-neighbourhood of every vertex $v\in V(G)$ of degree $d$ is the union of $k$-arrows, $k$-loops and non-arrows which are disjoint apart from their endpoints. Hence  there is a $\sigma$-structure $\struc{A}$ such that $f(\struc{A})\cong G$. Let $g$ be an isomorphism from $f(\struc{A})$ to  $G$.  
	
	Now we argue that $\struc{A}\in \classStruc{P}_{\rho_k}$. First assume that there are two elements $a,b$ with $(\mathcal{N}_2^{\struc{A}}(a),a)\in \tau_{d,2,\sigma}^{k}$ and $(\mathcal{N}_2^{\struc{A}}(b),b)\in \tau_{d,2,\sigma}^{k}$. By definition, we get that $(\mathcal{N}_{4\ell+2}^{f(\struc{A})}(a),a)\in\tau_{d,4\ell+2}^{k}$ and $(\mathcal{N}_{4\ell+2}^{f(\struc{A})}(b),b)\in\tau_{d,4\ell+2}^{k}$. Since $g$ is an isomorphism, the restriction of $g$ to $N_{4\ell+2}^{f(\struc{A})}(a)$ must be an isomorphism from $\mathcal{N}_{4\ell+2}^{f(\struc{A})}(a)$ to $\mathcal{N}_{4\ell+2}^{G}(g(a))$, and hence $(\mathcal{N}_{4\ell+2}^{G}(g(a)),g(a))\cong (\mathcal{N}_{4\ell+2}^{f(\struc{A})}(a),a)\in \tau_{d,4\ell+2}^{k}$. But the same holds for the $(4\ell+2)$-ball of $g(b)$, and hence we contradict the assumption that $G\in \classStruc{P}_{\hat{\rho}_k}$ since $\hat{\rho}_k(k)=[0,1]$.
	Let us further assume that there is an $a\in \univ{A}$ such that $(\mathcal{N}_2^{\struc{A}}(a),a)\in \tau_{d,2,\sigma}^i$ for some $i\notin I_{k}$. Let $j$ be the index such that $(\mathcal{N}_{4\ell+2}^{f(\struc{A})}(a),a)\in \tau_{d,4\ell+2}^j$. Additionally note that $a$ must have degree $d$ in $f(\struc{A})$ by construction of $f$.  As $g$ is an isomorphism,  we get that $(\mathcal{N}_{4\ell+2}^G(g(a)),g(a))\in \tau_{d,4\ell+2}^j$, and $g(a)$ has degree $d$. But then by construction of $\hat{\rho}_k$, there must be $G'\in \graphProp$, and a vertex $v\in V(G')$ of degree $d$ such that $(\mathcal{N}_{4\ell+2}^{G'}(v),v)\in \tau_{d,4\ell+2}^j$. By construction of $\graphProp$, there is a structure  $\struc{A}\in \classStruc{P}_{\zigzag}'$ such that $f(\struc{A}')=G'$. Since $v$ has degree $d$, it must be an element in $\struc{A}'$. Furthermore $(\mathcal{N}_2^{\struc{A}'}(v),v)\in \tau_{d,2,\sigma}^i$  by choice of $i$ and $j$. Hence  $\struc{A}'\notin \classStruc{P}_{\rho_k}$. But this contradicts Lemma~\ref{lem:neighbouhoodProfilOfPZigZag}. 
	
	Hence we have shown that $\struc{A}\in P_{\rho_k}$. Then by Lemma~\ref{lem:neighbouhoodProfilOfPZigZag} $\struc{A}\in \classStruc{P}_{\zigzag}'$, and by construction $G\in \graphProp$. 
\end{proof}

Since by Claim~\ref{claim:profilOfGraphProperty} we can express $\graphProp$ as a finite union of properties each defined by a $0$-profile,  Theorem~\ref{thm:subsetOfFOIsLocal} implies that $\graphProp$ is GSF-local.
\end{proof}
\subsection{Putting everything together}
Now we prove our main theorem.

\begin{proof}[Proof of Theorem~\ref{thm:existenceLocalNonTestableProperty}]
Let the property $\classStruc{P}_{\zigzag}'$ of relational structures be as defined above. 
	Note that  $\classStruc{P}_{\zigzag}'$ is not testable, as $\classStruc{P}_{\zigzag}$ is not testable \cite[Theorem 4.4]{testingFO} and $\classStruc{P}_{\zigzag}'$ only differs from $\classStruc{P}_{\zigzag}$ by the empty structure. By Lemma \ref{lem:local_reduction} and Lemma \ref{lem:localReduction}, the graph property $\graphProp$ that is locally reduced from $\classStruc{P}_{\zigzag}'$ is not testable. Lemma \ref{lemma:graphproperty_gsf_local} shows that $\graphProp$ is also a GSF-local property. Hence there exists a GSF-local property of bounded-degree graphs which is not testable. 
	Furthermore, since having a POT implies being testable, this proves that there is a GSF-local property which has no POT. By Theorem~\ref{thm:charOfPOT} this implies that not all GSF-local 
	 properties are non-propagating.
\end{proof}

\bibliographystyle{alpha}
\bibliography{GSF-Locality}
\newpage
\appendix
\begin{center}\huge\bf Appendix \end{center}

\section{Formal definitions of property testers and POTs }\label{app:preliminaries}
Now we give the formal definitions of standard property testing and proximity-oblivious testing.
\begin{definition}[(Standard) property testing]
	Let $\mathcal{P}=\cup_{n\in \mathbb{N}}\mathcal{P}_{n}$ be a property. An \emph{$\epsilon$-tester} for $\mathcal{P}_n$ is a probabilistic algorithm which, given query access to a structure $\struc{A}\in \mathcal{C}$ with $n$ vertices/elements, 
	\begin{itemize}
		\item accepts $\struc{A}$ with probability $2/3$ if $\struc{A}\in \mathcal{P}_n$.
		\item rejects $\struc{A}$ with probability $2/3$ if $\struc{A}$ is $\epsilon$-far from $\mathcal{P}_n$.
	\end{itemize} 
	We say that a property $\mathcal{P}$ is \emph{testable} if for every $n\in \mathbb{N}$ and $\epsilon\in (0,1)$, there exists an $\epsilon$-tester for $\mathcal{P}_n$ that makes at most $q=q(\epsilon,d)$ queries. We say the property $\mathcal{P}$ is testable with \emph{one-sided error} if the $\epsilon$-tester always accepts $\struc{A}$ if $\struc{A}\in \mathcal{P}$. 
\end{definition}
We introduce below the formal definition of proximity-oblivious testers. 
\begin{definition}[(One-sided error) proximity-oblivious testing]
	Let $\mathcal{P}=\cup_{n\in \mathbb{N}}\mathcal{P}_{n}$ be a property. Let $\eta:(0,1]\to (0,1]$ 
	be a monotone function. A \emph{proximity-oblivious tester (POT)} with detection probability $\eta$ for $\mathcal{P}_n$ is a probabilistic algorithm which, given query access to a structure $\struc{A}\in \mathcal{C}$ with $n$ vertices/elements,
	\begin{itemize}
		\item accepts $\struc{A}$ with probability $1$ if $\struc{A}\in \mathcal{P}_n$.
		\item rejects $\struc{A}$ with probability at least $\eta(\dist(\struc{A},\mathcal{P}_n))$ if $\struc{A}\notin\mathcal{P}_n$, where $\dist(\struc{A},\mathcal{P}_n)$ is the minimum fraction of different edges between $\struc{A}$ and any other $\struc{A}'\in\mathcal{P}_n$. 
	\end{itemize} 
	We say that a property $\mathcal{P}$ is \emph{proximity-oblivious testable} if for every $n\in \mathbb{N}$, there exists a POT for $\mathcal{P}_n$ of constant query complexity with detection probability $\eta$. 
\end{definition}

\section{The FO formula}\label{app:A}
For the construction of the formula $\varphi_{\zigzag}$ we use a recursively defined sequence $(G_m)_{m\in \mathbb{N}_{>0}}$ of edge expanders \cite[Proposition 9.2]{Hoory06expandergraphs}. Using this sequence we define the formula $\varphi_{\zigzag}$ in such a way that any model restricted to relation $F$ forms a rooted complete $D^4$-ary tree. Furthermore, the formula enforces that restricted to the vertices of level $i$ of the tree the relation $E$ encodes the rotation map  of the expander $G_i$. The formula $\varphi_{\zigzag}$ is the conjunction of the following formulas.  For a more detailed explanation and a proof of the precise form of the models of $\varphi_{\zigzag}$ see \cite{testingFO}.\\ 

We use the following formula 
\begin{displaymath}
	\varphi_{\operatorname{root}}(x):= \forall y \lnot F(y,x),
\end{displaymath}
 which expresses that vertex $x$ is a root vertex, \ie has no incoming $F$-edges. We then define the formula $\varphi_{\operatorname{tree}}$ which expresses that the structure restricted to the relation $F$ locally looks like a tree. More precisely, the formula expresses that there is precisely one root vertex, that every other vertex has one incoming $F$-edge and every vertex either has no $F$-children or has precisely $D^4$ $F$-children. We furthermore attach a $R$-self-loop to the root and $D^4$ $L$-self-loops to the leaves. This was important in \cite{testingFO} to make structures degree regular, but is of no relevance to this proof.  
\begin{align*}
\varphi_{\operatorname{tree}}&:= \exists^{=1} x \varphi_{\operatorname{root}}(x)\land \forall x \Big(\big(\varphi_{\operatorname{root}}(x)\land R(x,x)\big)\lor 
\big(\exists^{=1} y F(y,x)\land \lnot \exists y R(x,y)\land \lnot \exists y R(y,x)\big)\Big)\land
\\&\forall x \bigg(\Big[\neg\exists y F(x,y)\land \bigwedge_{k\in \indexSetH} L_k(x,x)\land \forall y \big(y\not= x \rightarrow \bigwedge_{k\in \indexSetH}\lnot L_k(x,y) \wedge \bigwedge_{k\in \indexSetH}\lnot L_k(y,x)\big)\Big]
\\&\lor \Big[\lnot\exists y \bigvee_{k\in \indexSetH}\big(L_k(x,y)\lor L_k(y,x)\big)\land 
\\&
\bigwedge_{k\in \indexSetH}\exists y_{k} \Big(x\not=y_{k}\land F_{k}(x,y_{k})\land 
(\bigwedge_{k'\in \indexSetH,k'\not=k}\lnot F_{k'}(x,y_k))\land \forall y(y\not=y_k\rightarrow \lnot F_{k}(x,y))\Big)\Big]\bigg).
\end{align*}
We define formula $\varphi_{\operatorname{rotationMap}}$ which expresses that the edge relations restricted to the relations $E$ encode a rotation map.
\begin{align*}
&\varphi_{\operatorname{rotationMap}}:= \forall x \forall y \Big(\bigwedge_{i,j\in \indexSetRotation}(E_{i,j}(x,y)\rightarrow E_{j,i}(y,x))\Big)
\land
\\&\forall x \Big(\bigwedge_{i\in \indexSetRotation}\Big(\bigvee_{j\in \indexSetRotation}\big(\exists^{=1}y E_{i,j}(x,y)\land \bigwedge_{\substack{j'\in \indexSetRotation\\ j'\not=j}}\lnot \exists y E_{i,j'}(x,y)\big)\Big)\Big)
\end{align*}
The formula $\varphi_{\operatorname{base}}$ expresses that the children of the root vertex form the basis of the recursive construction of expanders. The basis of the recursive construction is the square of some
$D$ regular graph $H$ on $D^4$ vertices with edge expansion ratio $1/4$. Explicit constructions of graphs with such properties are given in \cite{ZigZagProductIntroduction}. We assume that this graph is given by a rotation map $\rot_H$, which is an encoding of $H$.
\begin{align*}
\varphi_{\operatorname{base}}:=&\forall x  \Big(\varphi_{\operatorname{root}}(x)\rightarrow\Big[\bigwedge_{i,j\in \indexSetRotation}\Big(E_{i,j}(x,x)\land \forall y \Big(x\not= y\rightarrow \big(\lnot E_{i,j}(x,y)\land \lnot E_{i,j}(y,x)\big)\Big)\Big)\land \\& \bigwedge_{\substack{ \rot_{H^2}(k,i)=(k',i')\\k,k'\in \indexSetH\\i,i'\in \indexSetRotation }}\exists y \exists y'\big(F_k(x,y)\land F_{k'}(x,y')\land E_{i,i'}(y,y')\big)\Big]\Big)
\end{align*}
We define the formula $\varphi_{\operatorname{recursion}}$ which expresses the recursive construction of the sequence $(G_m)_{m\in \mathbb{N}_{>0}}$. This formula also depends on the base graph $H$.
\begin{align*}
&\varphi_{\operatorname{recursion}}:= \forall x \forall z\bigg[\Big(\lnot \exists y F(x,y)\land \lnot \exists y F(z,y)\Big)\lor
\bigwedge_{\substack{k_1',k_2'\in \indexSetRotation\\\ell_1',\ell_2'\in \indexSetRotation}}\Big(\exists y \big[E_{k_1',\ell_1'}(x,y)\land E_{k_2',\ell_2'}(y,z)\big]\rightarrow \\&
\bigwedge_{\substack{i,j,i',j'\in [D], k,\ell\in \indexSetH\\\rot_H(k,i)=((k_1', k_2'),i')\\ \rot_H((\ell_2', \ell_1'),j)=(\ell,j')}}\exists x'\exists z'\big[ F_k(x,x')\land F_\ell(z,z')\land
E_{(i,j),(j',i')}(x',z')\big]\Big)\bigg]
\end{align*}
\section{Deferred proofs from Section~\ref{sec:relatingNotionsOfLocality}}\label{app:C}
\begin{proof}
	For the first direction assume $\varphi$ is an FO-sentence. Then by Hanf's Theorem (Theorem~\ref{thm:HNF}) there is a sentence $\psi$ in Hanf normal form such that $\classStruc{P}_\varphi=\classStruc{P}_\psi$. 
	
	We will first convert $\psi$ into a sentence in Hanf normal form where every Hanf sentence appearing has the same locality radius. Let $r\in \mathbb{N}$ be the maximum locality radius appearing in $\psi$, and let $\varphi^{\geq m}_\tau:=\exists ^{\geq m} x \phi_{\tau}(x)$ be a Hanf sentence, where $\tau$ is an $r'$-type for some $r'\leq r$. Let $\tau_1,\dots,\tau_k$ be a list of all  $r$-types of bounded degree $d$ for which $(\mathcal{N}_{r'}^{\struc{B}}(b),b)\in \tau$ for  $(\struc{B},b)\in \tau_i$, for every $1\leq i\leq k$. Let $\Pi$ be the set of all partitions of $m$ into $k$ parts. Let 
	\begin{displaymath}
	\tilde{\varphi}^{\geq m}_\tau :=\bigvee_{(m_1,\dots,m_k)\in \Pi}\phantom{ii}\bigwedge_{i=1}^k \exists ^{\geq m_i} x \phi_{\tau_i}(x).
	\end{displaymath} 
	\begin{claim}\label{claim:increasingRadius}
		$\varphi^{\geq m}_\tau$ is $d$-equivalent to $\tilde{\varphi}^{\geq m}_\tau$.
	\end{claim}
	\begin{proof}
		Assume that $\struc{A}\in \classStruc{C}_{d}$ satisfies $\varphi^{\geq m}_\tau$, and assume that $a_1,\dots,a_m$ are $m$ distinct elements with $(\mathcal{N}_{r'}^{\struc{A}}(a_j),a_j)\in \tau$, for every $1\leq j\leq m$. Let $\tilde{\tau}_j$ be the $r$-type for which $(\mathcal{N}_r^{\struc{A}}(a_j),a_j)\in \tilde{\tau}_j$. By choice of $\tau_1,\dots,\tau_k$, we get that there are indices $i_1,\dots,i_m$ such that $\tilde{\tau}_j=\tau_{i_j}$. For $i\in \{1,\dots,k\}$ let $m_i=|\{j\in \{1,\dots,m\}\mid i_j=i \}|$. Hence $\struc{A}\models \bigwedge_{i=1}^k \exists ^{\geq m_i} x \phi_{\tau_i}(x)$ and since additionally $(m_1,\dots,m_k)\in\Pi$ this implies $\struc{A}\models \tilde{\varphi}^{\geq m}_\tau$.
		
		On the other hand, let $\struc{A}\in \classStruc{C}_{d}$ satisfy $\tilde{\varphi}^{\geq m}_\tau$, and let $(m_1,\dots,m_k)\in \Pi$ be a partition of $m$ such that $\struc{A}\models \bigwedge_{i=1}^k \exists ^{\geq m_i} x \phi_{\tau_i}(x)$. For every $1\leq i\leq k$, let  $a_{1}^i,\dots, a_{m_i}^i$  be $m_i$ distinct elements such that $(\mathcal{N}_r^{\struc{A}}(a_{j}^i),a_{j}^i)\in \tau_i$, for every $1\leq j\leq m_i$. By choice of $\tau_1,\dots,\tau_k$, we get that $(\mathcal{N}_{r'}^{\struc{A}}(a_{j}^i),a_{j}^i)\in \tau$, for every pair $1\leq i\leq k$, $1\leq j\leq m_i$. But since $m_1+\dots+m_k=m$ this implies that $\struc{A}\models\varphi^{\geq m}_\tau$. This proves that $\varphi^{\geq m}_\tau$ and $\tilde{\varphi}^{\geq m}_\tau$ are $d$-equivalent. 
	\end{proof}
	Let $\psi'$ be the formula in which every Hanf-sentence $\varphi^{\geq m}_\tau$ for which $\tau$ is an $r'$-type for some $r'<r$ gets replaced by $\tilde{\varphi}^{\geq m}_\tau$. By a simple inductive argument using Claim~\ref{claim:increasingRadius}, we get that $\psi$ is $d$-equivalent to $\psi'$, and hence $\classStruc{P}_\varphi=\classStruc{P}_{\psi}=\classStruc{P}_{\psi'}$. Furthermore since $\tilde{\varphi}^{\geq m}_\tau$ is a Boolean combination of Hanf-sentences for every $\varphi^{\geq m}_\tau$, and any Boolean combination of Boolean combinations is a Boolean combination itself, $\psi'$ is in Hanf normal form.  
	Furthermore, every Hanf-sentence appearing in $\psi'$ has locality radius $r$ by construction.
	
	Since any Boolean combination can be converted into disjunctive normal form (Lemma~\ref{lem:DNF}), we can assume that $\psi'$ is a disjunction of sentences $\xi$ of the form
	\begin{displaymath}
	\xi=\bigwedge_{j=1}^k \exists ^{\geq m_j} x \phi_{\tau_j}(x)\land \bigwedge_{j=k+1}^\ell \lnot \exists ^{\geq m_j} x \phi_{\tau_j}(x),
	\end{displaymath}
	where $\ell\in \mathbb{N}_{\geq 1}$, $1\leq k \leq \ell$, $m_i\in \mathbb{N}_{\geq 1}$ and $\tau_i$ is an $r$-type for every $1\leq i\leq \ell$. We can further assume that every sentence in the disjunction $\psi'$ is satisfiable by some $\struc{A}\in \classStruc{C}_{d}$, as any sentence with no  bounded degree $d$ model can be removed from $\psi'$.
	
	Let $\tilde{\tau}_1,\dots,\tilde{\tau}_t$ be a list of all $r$-types of bounded degree $d$ in the order we fixed.
	Let $k_i:=\max(\{m_j \mid 1\leq j\leq k, \tau_j=\tilde{\tau}_i \}\cup\{0\})$ and $\ell_i:=\min(\{m_j \mid k+1\leq j\leq \ell, \tau_j=\tilde{\tau}_i \}\cup\{\infty\})$ for every $i\in \{1,\dots,t\}$. Since $\xi$ has at least one bounded degree model $k_i\leq \ell_i$ for every $i\in \{1,\dots,t\}$.
	Let $\rho: \{1,\dots,t\}\rightarrow \mathfrak{I}$ be the neighbourhood profile defined by
	$\rho(i):=[k_i,\ell_i]$ if $\ell_i<\infty$ and $\rho(i):=[k_i,\ell_i)$ otherwise. Then by construction, we get that $\classStruc{P}_\rho=\classStruc{P}_\xi$. Since $\psi'$ is a disjunction of formulas, each of which defines a property which can be defined by some neighbourhood profile, we get that $\classStruc{P}_{\psi'}$ must be a finite  union of properties defined by some neighbourhood profile. \\
	
	On the other hand, for every $r$-neighbourhood profile $\rho$ of degree $d$, $\tau_1,\dots,\tau_t$ a list of all $r$-types of bounded degree $d$ in the order fixed  and the formula
	\begin{displaymath}
	\varphi_\rho:=\bigwedge_{i\in \{1,\dots,t\},\atop{\rho(i)=[k_i,\ell_i]}}\Big(\exists ^{\geq k_i} x \phi_{\tau_i}(x)\land \lnot \exists ^{\geq \ell_i+1} x \phi_{\tau_i}(x)\Big)\land \bigwedge_{i\in \{1,\dots,t\},\atop{\rho(i)=[k_i,\infty)}}\exists ^{\geq k_i} x \phi_{\tau_i}(x)
	\end{displaymath}
	it clearly holds that $\classStruc{P}_\rho=\classStruc{P}_{\varphi_\rho}$. Hence every finite union of properties defined by neighbourhood profiles can be defined by the disjunction of the formulas $\varphi_\rho$ of all $\rho$ in the finite union.
\end{proof}
\section{Deferred proofs from Section~\ref{sec:proofMainThm}}\label{app:B}

\begin{proof}[Proof of Claim~\ref{claim:closedUnderUnion}]
	Let $\struc{A},\struc{A}'\in \classStruc{C}_{\sigma,d}$ such that $\struc{A}\models \tilde{\varphi}_{\zigzag}$ and $\struc{A}'\models\tilde{\varphi}_{\zigzag}$, where $\tilde{\varphi}_{\zigzag}$ was the formula obtained from $\varphi_{\zigzag}$ by removing the subformula $\exists^{=1} x \varphi_{\operatorname{root}}(x)$. Our aim is to prove that $\struc{A}\cup \struc{A}'\models \tilde{\varphi}_{\zigzag}$ where $\struc{A}\cup \struc{A}'$ denotes the disjoint union of $\struc{A}$ and $\struc{A}'$. For this we essentially prove that for any two elements $a\in \univ{A}$ and $b\in \univ{A'}$ the formula $\tilde{\varphi}_{\zigzag}$ does not require a tuple containing $a$ and $b$. \\
	
	Let us define formulas
	\begin{align*}
	&\varphi:= \forall x \Big(\big(\varphi_{\operatorname{root}}(x)\land R(x,x)\big)\lor  
	\big(\exists^{=1} y F(y,x)\land \lnot \exists y R(x,y)\land \lnot \exists y R(y,x)\big)\Big),
	\\&\psi(x):= \neg\exists y F(x,y)\land \bigwedge_{k\in \indexSetH} L_k(x,x)\land \forall y \Big(y\not= x \rightarrow \bigwedge_{k\in \indexSetH}\lnot L_k(x,y) \wedge \bigwedge_{k\in \indexSetH}\lnot L_k(y,x)\Big) \text{ and}
	\\& \chi(x):=\lnot\exists y \bigvee_{k\in \indexSetH}\big(L_k(x,y)\lor L_k(y,x)\big)\land 
	\\&
	\bigwedge_{k\in \indexSetH}\exists y_{k} \Big(x\not=y_{k}\land F_{k}(x,y_{k})\land 
	(\bigwedge_{k'\in \indexSetH,k'\not=k}\lnot F_{k'}(x,y_k))\land \forall y(y\not=y_k\rightarrow \lnot F_{k}(x,y))\Big).
	\end{align*}
	Then $\tilde{\varphi}_{\zigzag}:=  \varphi \land \forall x (\psi(x)\lor \chi(x))\land \varphi_{\operatorname{rotationMap}}\land \varphi_{\operatorname{base}}\land\varphi_{\operatorname{recursion}}$. Hence it is  sufficient to prove that  $\struc{A}\cup\struc{A}'\models \varphi$, $\struc{A}\cup\struc{A}'\models\forall x (\psi(x)\lor \chi(x))$, $\struc{A}\cup\struc{A}'\models\varphi_{\operatorname{rotationMap}}$, $\struc{A}\cup\struc{A}'\models\varphi_{\operatorname{base}}$ and $\struc{A}\cup\struc{A}'\models\varphi_{\operatorname{recursion}}$.\\
	
	We first argue that $\struc{A}\cup\struc{A}'\models \varphi$. Let $a\in \univ{A\cup A'}$ be arbitrary and assume without loss of generality that $a\in \univ{A}$. Assume that $\struc{A}\cup\struc{A}'\not\models \varphi_{\operatorname{root}}(a)\land R(a,a)$. Since  $\varphi_{\operatorname{root}}(x):= \forall y \lnot F(y,x)$ this implies that $\struc{A}\not\models \varphi_{\operatorname{root}}(a)\land R(a,a)$. Since $\struc{A}\models \varphi$ we get that $\struc{A}\models \exists^{=1} y F(y,a)\land \lnot \exists y R(a,y)\land \lnot \exists y R(y,a)$. Hence there is an element $b\in \univ{A}$ such that $(b,a)\in \rel{F}{\struc{A}}$. Furthermore, for every $b'\in \univ{A}$, $b'\not=b$ we have $(b',a)\notin \rel{F}{\struc{A}}$, $(a,b')\notin \rel{R}{\struc{A}}$ and $(b',a)\notin\rel{R}{\struc{A}}$. But because  $a$ cannot be in a tuple with any element in $\univ{A'}$ we get that $\struc{A}\cup \struc{A}'\models \exists^{=1} y F(y,a)\land \lnot \exists y R(a,y)\land \lnot \exists y R(y,a)$. Hence $\struc{A}\cup\struc{A}'\models \varphi$.  \\
	
	Next we prove that $\struc{A}\cup\struc{A}'\models\forall x (\psi(x)\lor \chi(x))$. Let $a\in \univ{A\cup A'}$ be arbitrary and assume without loss of generality that $a\in \univ{A}$. First assume that $(a,b)\notin \rel{F}{\struc{A}\cup\struc{A}'}$ for every $b\in \univ{A\cup A'}$. Since $\struc{A}$ is a substructure of $\struc{A}\cup \struc{A}'$ this means that $\struc{A}\models \neg\exists y F(a,y)$. But then $\struc{A}\not\models \bigwedge_{k\in \indexSetH}\exists y_{k} \Big(a\not=y_{k}\land F_{k}(a,y_{k})\Big)$ which implies $\struc{A}\not\models\chi(a)$. Since $\struc{A}\models \forall x (\psi(x)\lor \chi(x))$ this implies that $\struc{A}\models \psi(a)$. Hence for every $k\in \indexSetH$ we have $(a,a)\in \rel{L_k}{\struc{A}}$ and for every $b\in \univ{A}$, $b\not=a$ we have $(a,b),(b,a)\notin\rel{F_k}{\struc{A}}$. Since there are no tuples containing both elements from $\struc{A}$ and $\struc{A}'$ this directly implies that $\struc{A}\cup\struc{A}'\models \psi(a)$. 
	
	On the other hand, assume that there is $b\in \univ{A\cup A'}$ such that $(a,b)\in \rel{F}{\struc{A}\cup\struc{A}'}$. Since we are considering the disjoint union of $\struc{A}$ and $\struc{A}' $ this implies that $b$ must be an element from $\struc{A}$. Hence $\struc{A}\not\models\psi(a)$. Since $\struc{A}\models \forall x (\psi(x)\lor \chi(x))$ this implies that $\struc{A}\models \chi(a)$. Then for every $k\in \indexSetH$ there is an element $b\in \univ{A}$ such that $(a,b)\in \rel{F_k}{\struc{A}}$, $(a,b)\notin\rel{F_{k'}}{\struc{A}}$ for every $k'\in \indexSetH$, $k'\not=k$ and $(a,b')\notin\rel{F_k}{\struc{A}}$ for every $b'\in \univ{A}$, $b'\not=b$. But since in $\struc{A}\cup\struc{A}'$ there are no tuples containing both elements from $\struc{A}$ and $\struc{A}'$ this implies that $\struc{A}\cup\struc{A}'\models \chi(a)$. In conclusion we proved that $\struc{A}\cup\struc{A'}\models \forall x (\psi(x)\lor \chi(x))$.\\
	
	We now prove $\struc{A}\cup \struc{A}'\models\varphi_{\operatorname{rotationMap}}$. Hence assume $a,b\in \univ{A\cup A'}$ are arbitrary elements. First consider the case that $a,b$ are either both from $\univ{A}$ or both from $\univ{A'}$.  Then if for some $i,j\in \indexSetRotation$ we have that $(a,b)\in \rel{E_{i,j}}{\struc{A}\cup \struc{A}'}$ then $(b,a)\in \rel{E_{j,i}}{\struc{A}\cup \struc{A}'}$ because $\struc{A}\models \varphi_{\operatorname{rotationMap}}$ and $\struc{A}'\models \varphi_{\operatorname{rotationMap}}$. Now consider the case that $|\{a,b\}\cap \univ{A}|=1$. Then $(a,b)\notin \rel{E_{i,j}}{\struc{A}\cup \struc{A}'}$ and $(b,a)\notin \rel{E_{j,i}}{\struc{A}\cup \struc{A}'}$ and hence $\struc{A}\cup \struc{A'}\models \bigwedge_{i,j\in \indexSetRotation}(E_{i,j}(a,b)\rightarrow E_{j,i}(b,a))$. Therefore $\struc{A}\cup \struc{A}'\models  \forall x \forall y \Big(\bigwedge_{i,j\in \indexSetRotation}(E_{i,j}(x,y)\rightarrow E_{j,i}(y,x))\Big)$.
	
	Now consider an arbitrary element $a\in \univ{A\cup A'}$ and any $i\in \indexSetRotation$. Without loss of generality assume $a\in \univ{A}$. Since $\struc{A}\models \varphi_{\operatorname{rotationMap}}$ there must be an index $j\in \indexSetRotation$ and an element $b\in \univ{A}$ such that $(a,b)\in \rel{E_{i,j}}{\struc{A}}$. Furthermore, for every $b'\in \univ{A}$, $b'\not=b$ we have $(a,b')\notin \rel{E_{i,j}}{\struc{A}}$ and for every $j'\in \indexSetRotation$, $j'\not=j$ and every $\tilde{b}\in \univ{A}$ we have $(a,\tilde{b})\notin \rel{E_{i,j'}}{\struc{A}}$. But since $a\in \univ{A}$ it also holds that $(a,b')\notin \rel{E_{i,j'}}{\struc{A}}$ for every $b'\in\univ{A'}$ and every $j'\in \indexSetRotation$. Hence $\struc{A}\cup\struc{A'}\models \bigvee_{j\in \indexSetRotation}\big(\exists^{=1}y E_{i,j}(a,y)\land \bigwedge_{\substack{j'\in \indexSetRotation\\ j'\not=j}}\lnot \exists y E_{i,j'}(a,y)\big)$. This concludes the proof of $\struc{A}\cup \struc{A}'\models\varphi_{\operatorname{rotationMap}}$.\\
	
	We now prove $\struc{A}\cup \struc{A}'\models\varphi_{\operatorname{base}}$. Assume $a\in \univ{A\cup A'}$ is an arbitrary element such that $\struc{A}\cup \struc{A}'\models \varphi_{\operatorname{root}}(a)$. Without loss of generality assume $a\in \univ{A}$. Since $\varphi_{\operatorname{root}}(x):= \forall y \lnot F(y,x)$ and $\struc{A}\cup \struc{A}'\models \varphi_{\operatorname{root}}(a)$  we get that $\struc{A}\models \varphi_{\operatorname{root}}(a)$. Since $\struc{A}\models \varphi_{\operatorname{base}}$ this means that for every $i,j\in \indexSetRotation $ we have $(a,a)\in \rel{E_{i,j}}{\struc{A}}$ and $(a,b),(b,a)\notin \rel{E_{i,j}}{\struc{A}}$ for every $b\in \univ{A}$, $b\not=a$. Since further $(a,b),(b,a)\notin   \rel{E_{i,j}}{\struc{A}\cup\struc{A}'}$ for every $b\in \univ{A'}$ this implies that $\struc{A}\cup \struc{A}'\models \bigwedge_{i,j\in \indexSetRotation}\Big(E_{i,j}(a,a)\land \forall y \Big(a\not= y\rightarrow \big(\lnot E_{i,j}(a,y)\land \lnot E_{i,j}(y,a)\big)\Big)\Big)$. Furthermore, since $\struc{A}\models \varphi_{\operatorname{base}}$  and $\struc{A}\models \varphi_{\operatorname{root}}(a)$ for every $ k,k'\in \indexSetH, i,i'\in \indexSetRotation $ for which $\rot_{H^2}(k,i)=(k',i')$ there are $b,b'\in \univ{A}$ such that $(a,b)\in \rel{F_k}{\struc{A}}$, $(a,b')\in \rel{F_{k'}}{\struc{A}}$ and $(b,b')\in \rel{E_{i,i'}}{\struc{A}}$. Since $\struc{A}$ is a substructure of $\struc{A}\cup\struc{A}'$ this proves that $\struc{A}\cup\struc{A}\models \varphi_{\operatorname{base}}$.\\
	
	Finally we prove $\struc{A}\cup \struc{A}'\models \varphi_{\operatorname{recursion}}$. Hence assume $a,c\in \univ{A\cup A'}$ are arbitrary elements. Assume $\struc{A}\cup\struc{A}'\not\models \lnot \exists y F(a,y)\land \lnot \exists y F(c,y)$ and assume without loss of generality that there is $\tilde{a}\in \univ{A\cup A'}$ such that $(a,\tilde{a})\in \rel{F}{\struc{A}\cup\struc{A'}}$. Since there are no tuples containing both elements from $\struc{A}$ and $\struc{A}'$ we get that $a,\tilde{a}$ are from the same structure. Without loss of generality assume $a,\tilde{a}\in \univ{A}$. Assume that for indices $k_1',k_2'\in \indexSetRotation, \ell_1',\ell_2'\in \indexSetRotation$ and some element $b\in \univ{A\cup A'}$ we have $(a,b)\in \rel{E_{k_1',\ell_1'}}{\struc{A}\cup\struc{A}'}$ and $(b,c)\in  \rel{E_{k_2',\ell_2'}}{\struc{A}\cup\struc{A}'}$. As $b$ also has to be in $\univ{A}$ and  $\struc{A}\models \varphi_{\operatorname{recursion}}$ this implies that for every $i,j,i',j'\in [D], k,\ell\in \indexSetH$ for which $\rot_H(k,i)=((k_1', k_2'),i')$ and $\rot_H((\ell_2', \ell_1'),j)=(\ell,j')$ there are elements $a',c'\in \univ{A\cup A'}$ such that $(a,a')\in \rel{F_k}{\struc{A}\cup\struc{A}'}$, $(c,c')\in \rel{F_\ell}{\struc{A}\cup\struc{A}'}$ and $(a',c')\in \rel{E_{(i,j),(j',i')}}{\struc{A}\cup\struc{A}'}$. Hence $\struc{A}\cup \struc{A}'\models \varphi_{\operatorname{recursion}}$. 
\end{proof}
\begin{claim}\label{claim:satisfyingTree}
	Every structure $\struc{A}\in \bigcup_{1\leq k \leq m}\classStruc{P}_{\rho_k}\setminus \{\struc{A}_{\emptyset}\}$ satisfies $\varphi'_{\operatorname{tree}}$.
\end{claim}
\begin{proof}
	Let $\struc{A}\in \bigcup_{1\leq k \leq m}\classStruc{P}_{\rho_k}\setminus \{\struc{A}_{\emptyset}\}$. Then there is $k\in \{1,\dots,m\}$ such that $\struc{A}\in \classStruc{P}_{\rho_k}$.

	By definition, $\varphi'_{\operatorname{tree}}:= \exists^{\leq 1} x \varphi_{\operatorname{root}}(x)\land \varphi \land \forall x (\psi(x)\lor \chi(x))$,  where
	\begin{align*}
	&\varphi:= \forall x \Big(\big(\varphi_{\operatorname{root}}(x)\land R(x,x)\big)\lor  
	\big(\exists^{=1} y F(y,x)\land \lnot \exists y R(x,y)\land \lnot \exists y R(y,x)\big)\Big),
	\\&\psi(x):= \neg\exists y F(x,y)\land \bigwedge_{k\in \indexSetH} L_k(x,x)\land \forall y \Big(y\not= x \rightarrow \bigwedge_{k\in \indexSetH}\lnot L_k(x,y) \wedge \bigwedge_{k\in \indexSetH}\lnot L_k(y,x)\Big) \text{ and}
	\\& \chi(x):=\lnot\exists y \bigvee_{k\in \indexSetH}\big(L_k(x,y)\lor L_k(y,x)\big)\land 
	\\&
	\bigwedge_{k\in \indexSetH}\exists y_{k} \Big(x\not=y_{k}\land F_{k}(x,y_{k})\land 
	(\bigwedge_{k'\in \indexSetH,k'\not=k}\lnot F_{k'}(x,y_k))\land \forall y(y\not=y_k\rightarrow \lnot F_{k}(x,y))\Big).
	\end{align*}
	Thus, it is sufficient to prove that $\struc{A}\models \exists^{\leq 1} x \varphi_{\operatorname{root}}(x)$, $\struc{A}\models \varphi$ and $\struc{A}\models\forall x (\psi(x)\lor \chi(x))$. 
	
	To prove $\struc{A}\models \exists^{\leq 1} x \varphi_{\operatorname{root}}(x)$ we note that by construction of $\rho_k$ we have $\struc{A}\not\models \varphi_{\operatorname{root}}(a)$ 
	for any $a\in \univ{A}$ for which $(\mathcal{N}_2^{\struc{A}}(a),a)\notin \tau_{d,2,\sigma}^k$. Since $\rho_k$ restricts the number of occurrences of elements of neighbourhood type $\tau_{d,2,\sigma}^{k}$ to at most one, this proves that there is at most one $a\in \univ{A}$ with $\struc{A}\models \varphi_{\operatorname{tree}}(a)$ and hence $\struc{A}\models \exists^{\leq 1} x \varphi_{\operatorname{root}}(x)$. 
	
	To prove $\struc{A}\models \varphi$, let $a\in \univ{A}$ be an arbitrary element. 
	Since $\struc{A}\in \classStruc{P}_{\rho_k}$, there is an $i\in I_{k}$ such that $(\mathcal{N}_2^{\struc{A}}(a),a)\in \tau_{d,2,\sigma}^i$. 
	But then by definition, there exist $\other{\struc{A}}\models \varphi_{\zigzag}$ and $\other{a}\in \univ{\other{A}}$ such that $(\mathcal{N}_2^{\struc{A}}(a),a)\cong (\mathcal{N}_2^{\other{\struc{A}}}(\other{a}),\other{a})$. 
	Assume $f$ is an isomorphism from  $(\mathcal{N}_2^{\struc{A}}(a),a)$ to  $(\mathcal{N}_2^{\other{\struc{A}}}(\other{a}),\other{a})$. First consider the case that $\struc{A}\models \varphi_{\operatorname{root}}(a):=\forall y \lnot F(y,a)$. Assume there  is $\other{b}\in \univ{\other{A}}$ such that $(\other{b},\other{a})\in \rel{F}{\other{\struc{A}}}$. Since $\other{b}\in N_2^{\other{\struc{A}}}(\other{a})$, there must be an element $b\in N_2^{\struc{A}}(a)$ such that $f(b)=\other{b}$. Since $f$ is an isomorphism mapping $a$ to $\other{a}$, this implies $(b,a)\in \rel{F}{\struc{A}}$, which contradicts $\struc{A}\models \varphi_{\operatorname{root}}(a)$. 
	Hence $\other{\struc{A}}\models \varphi_{\operatorname{root}}(\other{a})$. Since $\other{\struc{A}}\models \varphi_{\operatorname{tree}}'$, it holds that $\other{\struc{A}}\models \varphi$, which means that $(\other{a},\other{a})\in \rel{R}{\other{\struc{A}}}$. But since $f$ is an isomorphism mapping $a$ onto $\other{a}$, this implies $(a,a)\in \rel{R}{\struc{A}}$. Now consider the case that  
	$\struc{A}\not\models \varphi_{\operatorname{root}}(a)$. Then there is $b\in \univ{A}$ with $(b,a)\in \rel{F}{\struc{A}}$. Since $f$ is an isomorphism, this implies $(f(b),\other{a})\in \rel{F}{\other{\struc{A}}}$. Hence $\other{\struc{A}}\models \exists^{=1} y F(y,\other{a})\land \lnot \exists y R(\other{a},y)\land \lnot \exists y R(y,\other{a})$, as $\other{\struc{A}}\models \varphi$. Now assume that  there is $b'\not=b$ such that $(b',a)\in \rel{F}{\struc{A}}$. Then $f(b)\not=f(b')$ and $(f(b),\other{a}),(f(b'),\other{a})\in \rel{F}{\other{\struc{A}}}$. Since this contradicts $\other{\struc{A}}\models \exists^{=1} y F(y,\other{a})$ we have $\struc{A}\models \exists^{=1} y F(y,a)$. Furthermore, assume that there is $b'\in \univ{A}$ such that either $(a,b')\in \rel{R}{\struc{A}}$ or $(b',a)\in \rel{R}{\struc{A}}$. Then either $(\other{a},f(b'))\in \rel{R}{\other{\struc{A}'}}$ or $(f(b'),\other{a})\in \rel{R}{\other{\struc{A}}}$, which contradicts $\other{\struc{A}}\models \lnot \exists R(\other{a},y)\land \lnot \exists y R(y,\other{a})$. Therefore $\struc{A}\models \lnot \exists R(a,y)\land \lnot \exists y R(y,a)$ which completes the proof of $\struc{A}\models\varphi$.  
	
	We prove $\struc{A}\models \forall x (\psi(x)\lor \chi(x))$ by considering the two cases $\struc{A}\models \neg\exists y F(a,y)$ and $\struc{A}\models \exists yF(a,y)$ for each element $a\in \univ{A}$. For this, let $a\in \univ{A}$ be any element. By the construction of $\rho_k$ there is $\other{\struc{A}}\models \varphi_{\zigzag}$ and $\other{a}\in \univ{\other{A}}$ such that $(\mathcal{N}_2^{\struc{A}}(a),a)\cong (\mathcal{N}_2^{\other{\struc{A}}}(\other{a}),\other{a})$. Let $f$ be an isomorphism from  $(\mathcal{N}_2^{\struc{A}}(a),a)$ to $(\mathcal{N}_2^{\other{\struc{A}}}(\other{a}),\other{a})$. First consider the case that $\struc{A}\models \neg\exists y F(a,y)$. If there was $\other{b}\in \univ{\other{A}}$ with $(\other{a},\other{b})\in \rel{F}{\other{\struc{A}}}$ then $(a,f^{-1}(\other{b}))\in \rel{F}{\struc{A}}$ contradicting our assumption. Hence $\other{\struc{A}}\models \neg\exists y F(\other{a},y)$  which implies that $\other{\struc{A}}\not\models \chi(\other{a})$. But since $\other{\struc{A}}\models \varphi_{\zigzag}$, it holds that $\other{\struc{A}}\models \forall x (\psi(x)\lor \chi(x))$, which implies that $\other{\struc{A}}\models \psi(\other{a})$. Hence $(\other{a},\other{a})\in \rel{L_k}{\other{\struc{A}}}$ for every $k\in \indexSetH$. Since $f$ is an isomorphism and $f(a)=\other{a}$, it holds that $(a,a)\in \rel{L_k}{\struc{A}}$ for every $k\in \indexSetH$, and hence $\struc{A}\models \bigwedge_{k\in \indexSetH} L_k(a,a)$. Furthermore, assume that there is $b\in \univ{A}$, $b\not=a$ and $k\in \indexSetH$ such that either $(a,b)\in \rel{L_k}{\struc{A}}$ or $(b,a)\in \rel{L_k}{\struc{A}}$. Since $f$ is an isomorphism this implies that either $(\other{a},f(b))\in \rel{L_k}{\other{\struc{A}}}$ or $(f(b),\other{a})\in \rel{L_k}{\other{\struc{A}}}$ which contradicts $\other{\struc{A}}\models \chi(\other{a})$. Hence $\struc{A}\models\forall y \Big(y\not= a \rightarrow \bigwedge_{k\in \indexSetH}\lnot L_k(a,y) \wedge \bigwedge_{k\in \indexSetH}\lnot L_k(y,a)\Big)$ proving that $\struc{A}\models \psi(a)$. 
	
	 Now consider the case that there is an element $b\in \univ{A}$ such that $(a,b)\in \rel{F}{\struc{A}}$. Since this implies that $(\other{a},f(b))\in \rel{F}{\other{\struc{A}}}$, we get that $\other{\struc{A}}\not\models \psi(\other{a})$, and hence $\other{\struc{A}}\models \chi(\other{a})$. Now assume that there is  $b\in \univ{A}$ and $k\in \indexSetH$ such that either $(a,b)\in \rel{L_k}{\struc{A}}$ or $(b,a)\in \rel{L_k}{\struc{A}}$. But then either $(\other{a},f(b))\in \rel{L_k}{\other{\struc{A}}}$ or $(f(b),\other{a})\in \rel{L_k}{\other{\struc{A}}}$, which contradicts $\other{\struc{A}}\models \chi(\other{a})$. Hence $\struc{A}\models  \lnot\exists y \bigvee_{k\in \indexSetH}\big(L_k(a,y)\lor L_k(y,a)\big)$. For each $k\in \indexSetH$, let $\other{b}_k\in \univ{\other{A}}$ be an element such that $\other{\struc{A}}\models \other{a}\not=\other{b}_{k}\land F_{k}(\other{a},\other{b}_{k})\land 
	(\bigwedge_{k'\in \indexSetH,k'\not=k}\lnot F_{k'}(\other{a},\other{b}_k))\land \forall y(y\not=\other{b}_k\rightarrow \lnot F_{k}(\other{a},y))$. Since $f$ is an isomorphism,  this implies that $a\not= b_k:=f^{-1}(\other{b}_k)$, $(a,b_k)\in \rel{F_k}{\struc{A}}$ and $(a,b_k)\notin \rel{F_{k'}}{\struc{A}}$, for each $k'\in \indexSetH,k'\not=k$. Furthermore, assume there is  $b\in \univ{A}$, $b\not=b_k$ such that $(a,b)\in \rel{F_k}{\struc{A}}$. Since $f$ is an isomorphism, this implies $f(b)\not=f(b_k)=\other{b}_k$ and $(\other{a},\other{b})\in \rel{F_k}{\other{\struc{A}}}$, which contradicts $\other{\struc{A}}\models \forall y(y\not=\other{b}_k\rightarrow \lnot F_{k}(\other{a},y))$. Hence $\struc{A}\models \forall y(y\not=b_k\rightarrow \lnot F_{k}(a,y))$ and therefore concluding that $\struc{A}\models \chi(a)$. This proves that in either case $\struc{A}\models \psi(a)\lor\chi(a)$ and therefore $\struc{A}\models \forall x(\psi(x)\lor \chi(x))$. 
\end{proof}
\begin{claim}\label{claim:satisfyingRotationMap}
	Every structure $\struc{A}\in \bigcup_{1\leq k \leq m}\classStruc{P}_{\rho_k}\setminus \{\struc
	{A}_{\emptyset}\}$ satisfies $\varphi_{\operatorname{rotationMap}}$.
\end{claim}
\begin{proof}
	Let $\struc{A}\in \bigcup_{1\leq k \leq m}\classStruc{P}_{\rho_k}\setminus \{\struc{A}_{\emptyset}\}$. Then there is a $k\in \{1,\dots,m\}$ such that $\struc{A}\in \classStruc{P}_{\rho_k}$.
	
	By definition, $\varphi_{\operatorname{rotationMap}}=\varphi\land \psi$, where 
	\begin{align*}
	&\varphi:= \forall x \forall y \Big(\bigwedge_{i,j\in \indexSetRotation}(E_{i,j}(x,y)\rightarrow E_{j,i}(y,x))\Big) \text{ and }\\
	&\psi:= \forall x \Big(\bigwedge_{i\in \indexSetRotation}\Big(\bigvee_{j\in \indexSetRotation}\big(\exists^{=1}y E_{i,j}(x,y)\land \bigwedge_{\substack{j'\in \indexSetRotation\\ j'\not=j}}\lnot \exists y E_{i,j'}(x,y)\big)\Big)\Big).
	\end{align*}
	Thus, it is sufficient to prove that $\struc{A}\models\varphi$ and $\struc{A}\models \psi$. 
	
	To prove $\struc{A}\models\varphi$, assume towards a contradiction that there are $a,b\in \univ{A}$ such that for some pair $i,j\in \indexSetRotation$, we have that $(a,b)\in \rel{E_{i,j}}{\struc{A}}$, but $(b,a)\notin \rel{E_{j,i}}{\struc{A}}$. By construction of  $ \classStruc{P}_{\rho_k}$, there is a structure $\other{\struc{A}}\models \varphi_{\zigzag}$ and $\other{a}\in \univ{\other{A}}$ such that $(\mathcal{N}_2^{\struc{A}}(a),a)\cong (\mathcal{N}_2^{\other{\struc{A}}}(\other{a}),\other{a})$. Assume $f$ is an isomorphism from $(\mathcal{N}_2^{\struc{A}}(a),a)$ to  $(\mathcal{N}_2^{\other{\struc{A}}}(\other{a}),\other{a})$. Note that $f(b)$ is defined since $b$ is in the $2$-neighbourhood of $a$. Furthermore since $f$ is an isomorphism, $(a,b)\in \rel{E_{i,j}}{\struc{A}}$ implies $(\other{a},f(b))\in \rel{E_{i,j}}{\other{\struc{A}}}$, and $(b,a)\notin \rel{E_{j,i}}{\struc{A}}$ implies $(f(b),\other{a})\notin \rel{E_{j,i}}{\other{\struc{A}}}$. Hence $\other{\struc{A}}\not\models \varphi$, which contradicts $\other{\struc{A}}\models \varphi_{\operatorname{rotationMap}}$. 
	
	To prove $\struc{A}\models\psi$, assume towards a contradiction that there is an $a\in \univ{A}$ and $i\in \indexSetRotation$ such that  $\struc{A}\not\models \exists^{=1}y E_{i,j}(a,y)\land \bigwedge_{\substack{j'\in \indexSetRotation\\ j'\not=j}}\lnot \exists y E_{i,j'}(a,y)$ for every  $j\in \indexSetRotation$. We know that there is a structure $\other{\struc{A}}\models \varphi_{\zigzag}$ and $\other{a}\in \univ{\other{A}}$ such that $(\mathcal{N}_2^{\struc{A}}(a),a)\cong (\mathcal{N}_2^{\other{\struc{A}}}(\other{a}),\other{a})$. Let $f$ be an isomorphism from $(\mathcal{N}_2^{\struc{A}}(a),a)$ to $(\mathcal{N}_2^{\other{\struc{A}}}(\other{a}),\other{a})$. Since $\other{\struc{A}}\models \psi$, there must be $j\in \indexSetRotation$ such that $\other{\struc{A}}\models \exists^{=1}y E_{i,j}(\other{a},y)\land \bigwedge_{\substack{j'\in \indexSetRotation\\ j'\not=j}}\lnot \exists y E_{i,j'}(\other{a},y)$. Hence there must be $\other{b}\in \univ{\other{A}}$ such that $(\other{a},\other{b})\in \rel{E_{i,j}}{\other{\struc{A}}}$, which implies that $(a,f^{-1}(\other{b}))\in \rel{E_{i,j}}{\struc{A}}$. Since we assumed that $\struc{A}\not\models \exists^{=1}y E_{i,j}(a,y)\land \bigwedge_{\substack{j'\in \indexSetRotation\\ j'\not=j}}\lnot \exists y E_{i,j'}(a,y)$, there must be either $b\not=f^{-1}(\other{b})$ with $(a,b)\in \rel{E_{i,j}}{\struc{A}}$, or there must be $j'\in \indexSetRotation$, $j'\not=j$ and $b'\in \univ{A}$ such that $(a,b')\in \rel{E_{i,j'}}{\struc{A}}$.  In the first case  $(\other{a},f(b))\in \rel{E_{i,j}}{\other{\struc{A}}}$, since $f$ is an isomorphism. But then $\other{\struc{A}}\not\models \exists^{=1}y E_{i,j}(\other{a},y)$, which is a contradiction. In the second case, we get that $(\other{a},f(b'))\in \rel{E_{i,j'}}{\other{\struc{A}}}$. But then $\other{\struc{A}}\not\models \bigwedge_{\substack{j'\in \indexSetRotation\\ j'\not=j}}\lnot \exists y E_{i,j'}(\other{a},y)$, which is a contradiction. Hence $\struc{A}\models \varphi \land \psi$.
\end{proof}
\begin{claim}\label{claim:satisfyingBase}
	Every structure $\struc{A}\in \bigcup_{1\leq k \leq m}\classStruc{P}_{\rho_k}\setminus \{\struc{A}_{\emptyset}\}$ satisfies $\varphi_{\operatorname{base}}$.
\end{claim}
\begin{proof}
	Let $\struc{A}\in \bigcup_{1\leq k \leq m}\classStruc{P}_{\rho_k}\setminus \{\struc{A}_{\emptyset}\}$. Then there is a $k\in \{1,\dots,m\}$ such that $\struc{A}\in \classStruc{P}_{\rho_k}$.
	
	By definition, $\varphi_{\operatorname{base}}:=\forall x  \big(\varphi_{\operatorname{root}}(x)\rightarrow (\varphi(x)\land\psi(x))\big)$, where 
	\begin{align*}
	\varphi(x):=&\bigwedge_{i,j\in \indexSetRotation}\Big(E_{i,j}(x,x)\land \forall y \Big(x\not= y\rightarrow \big(\lnot E_{i,j}(x,y)\land \lnot E_{i,j}(y,x)\big)\Big)\Big)\text{ and }\\
	\psi(x):=& \bigwedge_{\substack{ \rot_{H^2}(k,i)=(k',i')\\k,k'\in \indexSetH\\i,i'\in \indexSetRotation }}\exists y \exists y'\big(F_k(x,y)\land F_{k'}(x,y')\land E_{i,i'}(y,y')\big).
	\end{align*}
	Thus, it is sufficient to prove that $\struc{A}\models \varphi(a)$ and $\struc{A}\models \psi(a)$ for every $a\in \univ{A}$ for which $\struc{A}\models \varphi_{\operatorname{root}}(a)$. Therefore assume $a\in \univ{A}$ is any element such that $\struc{A}\models \varphi_{\operatorname{root}}(a)$. Because $\struc{A}\in \classStruc{P}_{\rho_k}$ there is an $i\in I_{k}$ such that $(\mathcal{N}_2^{\struc{A}}(a),a)\in \tau_{d,2,\sigma}^i$. Then by definition there is a structure $\other{\struc{A}}\models \varphi_{\zigzag}$ and $\other{a}\in \univ{\other{A}}$ such that $(\mathcal{N}_2^{\struc{A}}(a),a)\cong (\mathcal{N}_2^{\other{\struc{A}}}(\other{a}),\other{a})$.  Let $f$ be an isomorphism from $(\mathcal{N}_2^{\struc{A}}(a),a)$ to $(\mathcal{N}_2^{\other{\struc{A}}}(\other{a}),\other{a})$. Assume that there is an element $\other{b}\in \univ{\other{A}}$ such that $(\other{b},\other{a})\in \rel{F}{\other{\struc{A}}}$. Since $f$ is an isomorphism and $\other{b}\in N_2^{\other{\struc{A}}}(\other{a})$ we get that $(f^{-1}(\other{b}),a)\in \rel{F}{\struc{A}}$ which contradicts that $\struc{A}\models\varphi_{\operatorname{root}}(a)$ as $\varphi_{\operatorname{root}}(x):= \forall y \lnot F(y,x)$. Hence there is no element  $\other{b}\in \univ{\other{A}}$ such that $(\other{b},\other{a})\in \rel{F}{\other{\struc{A}}}$ which implies that $\other{\struc{A}}\models \varphi_{\operatorname{root}}(\other{a})$. But since $\other{\struc{A}}\models \varphi_{\zigzag}$ we have that $\other{\struc{A}}\models \varphi_{\operatorname{base}}$ and hence $\other{\struc{A}}\models \varphi(\other{a})$ and $\other{\struc{A}}\models \psi(\other{a})$.
	
	To prove $\struc{A}\models\varphi(a)$ first observe that $(a,a)\in \rel{E_{i,j}}{\struc{A}}$ for every $i,j\in \indexSetRotation$ since $\other{\struc{A}}\models \varphi(\other{a})$ implies that $(\other{a},\other{a})\in \rel{E_{i,j}}{\other{\struc{A}}}$ for every $i,j\in \indexSetRotation$ and $f$ is an isomorphism mapping $a$ onto $\other{a}$. Assume that there is an element $b\in \univ{A}$, $b\not=a$ and indices $i,j\in \indexSetRotation$ such that either $(a,b)\in \rel{E_{i,j}}{\struc{A}}$ or $(b,a)\in \rel{E_{i,j}}{\struc{A}}$. Since $b\in N_2^{\struc{A}}(a)$ and $f$ is an isomorphism we get that $f(b)\not=f(a)=\other{a}$ and either $(\other{a},f(b))\in \rel{E_{i,j}}{\other{\struc{A}}}$ or $(f(b),\other{a})\in \rel{E_{i,j}}{\other{\struc{A}}}$. But this contradicts $\other{\struc{A}}\models \varphi(\other{a})$ and hence $\struc{A}\models \varphi(a)$.
	
	We now prove $\struc{A}\models \psi(a)$. Let  $k,k'\in \indexSetH$ and $i,i'\in \indexSetRotation$ such that $\rot_{H^2}(k,i)=(k',i')$. Since $\other{\struc{A}}\models \psi(\other{a})$ there must be elements $\other{b},\other{b}'\in \univ{\other{A}}$ such that $(\other{a},\other{b})\in \rel{F_k}{\other{\struc{A}}}$, $(\other{a},\other{b}')\in \rel{F_{k'}}{\other{\struc{A}}}$ and $(\other{b},\other{b}')\in \rel{E_{i,i'}}{\other{\struc{A}}}$. But since $\other{b},\other{b}'\in N_2^{\other{\struc{A}}}(\other{a})$ we get that $f^{-1}(\other{b})$ and $f^{-1}(\other{b}')$ are defined  and since $f$ is an isomorphism we get that $(a,f^{-1}(\other{b}))\in \rel{F_k}{\struc{A}}$, $(a,f^{-1}(\other{b}'))\in \rel{F_{k'}}{\struc{A}}$ and $(f^{-1}(\other{b}),f^{-1}(\other{b}'))\in \rel{E_{i,i'}}{\struc{A}}$. Hence $\struc{A}\models \exists y \exists y'\big(F_k(a,y)\land F_{k'}(a,y')\land E_{i,i'}(y,y')$ for any $k,k'\in \indexSetH$ and $i,i'\in \indexSetRotation$ such that $\rot_{H^2}(k,i)=(k',i')$ which implies that $\struc{A}\models \psi(a)$.
\end{proof}
\end{document}